%% file: dper.tex
\documentclass[
    11pt, % 10 (default), 11, 12
]{article}
\usepackage[scale=.8]{geometry}

\date{}

\newcommand{\papertitle}{DPER: Dynamic Programming for Exist-Random Stochastic SAT}

\title{
    \papertitle%
    \footnote{Work supported in part by NSF grants IIS-1527668, CCF-1704883, IIS-1830549, and CNS-2016656; DoD MURI grant N00014-20-1-2787; and an award from the Maryland Procurement Office.}
}

\newcommand{\authors}{Vu H. N. Phan and Moshe Y. Vardi}

\author{
    \authors \\
    \texttt{\{vhp1,vardi\}@rice.edu} \\
    Rice University
}

\usepackage[
    pdftitle=\papertitle,
    pdfauthor=\authors,
    hidelinks,
]{hyperref}

\usepackage[commandnameprefix=ifneeded]{changes}
\definechangesauthor[name=Moshe, color=red]{myv}
\definechangesauthor[name=Vu, color=blue]{vhp}

\usepackage[none]{hyphenat}

\usepackage{breakcites}

\usepackage{amsthm}

\usepackage{thm-restate}

\newtheorem{corollary}{Corollary}
\newtheorem{definition}{Definition}
\newtheorem{lemma}{Lemma}

\usepackage{bookmark}
\bookmarksetup{
    numbered,
    open,
}

\usepackage{float}

\usepackage{multirow}

\usepackage{tikz-qtree}

\usepackage[
    ruled,
    linesnumbered,
]{algorithm2e}
\SetCommentSty{}
\SetKw{Assert}{assert}
\SetKwFor{For}{for}{}{}
\SetKwFor{While}{while}{}{}
\SetKwIF{If}{ElseIf}{Else}{if}{}{else if}{else}{}

\usepackage{amsmath}
\usepackage[
    capitalize,
    noabbrev,
]{cleveref}

\newcommand{\ms}{MaxSAT}

\newcommand{\cnfs}{849}
\newcommand{\plannedcnfs}{241}
\newcommand{\vbscnfs}{484}
\newcommand{\fastestcnfs}{62}
\newcommand{\fastestpercent}{13\%} % of \vbscnfs

\newcommand{\tool}[1]{\texttt{#1}}

\newcommand{\dcssat}{\tool{DC-SSAT}}
\newcommand{\dper}{\tool{DPER}}
\newcommand{\dpmc}{\tool{DPMC}}

\newcommand{\erssat}{\tool{erSSAT}}
\newcommand{\executor}{\tool{DMC}}

\newcommand{\maxcount}{\tool{MaxCount}}

\newcommand{\planner}{\tool{LG}}
\newcommand{\procount}{\tool{ProCount}}

\newcommand{\vbs}[1]{\tool{VBS{#1}}}

\newcommand{\angles}[1]{\langle #1 \rangle}
\newcommand{\braces}[1]{\left\{ #1 \right\}}
\newcommand{\brackets}[1]{\left[ #1 \right]}
\newcommand{\pars}[1]{\left( #1 \right)}
\newcommand{\pipes}[1]{\left| #1 \right|}

\newcommand{\tup}{\angles} % tuple

\newcommand{\set}{\braces}

\newcommand{\pb}[1]{\brackets{#1}} % pseudo-Boolean function represented by formula

\newcommand{\of}{\pars}

\newcommand{\abs}[1]{\pipes{#1}}
\newcommand{\size}[1]{\pipes{#1}}

\newcommand{\bigo}[1]{\operatorname{O}\of{#1}}

\newcommand{\vars}[1]{\operatorname{vars}\of{#1}}

\newcommand{\dom}[1]{\operatorname{dom}\of{#1}} % domain

\newcommand{\width}[1]{\operatorname{width}\of{#1}}

\newcommand{\func}[1]{\mathtt{#1}}

\newcommand{\valuator}{\func{Valuator}}

\newcommand{\push}[2]{\func{push}\of{#1, #2}}
\newcommand{\pop}[1]{\func{pop}\of{#1}}

\newcommand{\minsert}[2]{\func{insert}\of{#1, #2}} % multiset
\renewcommand{\minsert}[2]{#1 \gets #1 \cup \set{#2}}

\newcommand{\mremove}[2]{\func{remove}\of{#1, #2}} % multiset
\renewcommand{\mremove}[2]{#1 \gets #1 \setminus \set{#2}}

\newcommand{\ta}{\tau} % truth assignment

\newcommand{\va}[2]{\tup{#1, #2}} % var assignment

\newcommand{\extend}[3]{#1 \cup \set{\va{#2}{#3}}}

\newcommand{\pr}{\rho} % probability mapping

\usepackage{stmaryrd}
\newcommand{\val}[1]{{\llbracket #1 \rrbracket}_\pr} % valuation

\usepackage{amssymb}
\newcommand{\restrict}[2]{#1\restriction_{#2}}

\newcommand{\B}{\mathbb{B}}
\newcommand{\R}{\mathbb{R}}

\newcommand{\ps}[1]{\B^{#1}} % powerset

\renewcommand{\emptyset}{\varnothing}

\renewcommand{\phi}{\varphi}

\newcommand{\V}[1]{\mathcal{V}\of{#1}} % vertices(graph)

\newcommand{\C}[1]{\mathcal{C}\of{T, r, #1}} % children

\newcommand{\Lv}{\mathcal{L}\of{T, r}} % leaves

\newcommand{\T}{\mathcal{T}} % project-join tree

\newcommand{\I}{\mathcal{I}} % grade

\newcommand{\gammai}[1]{\gamma^{-1}\of{#1}} % gamma inverse

\newcommand{\gammamap}{\overset{\gamma}{\mapsto}}
\newcommand{\pimap}{\overset{\pi}{\mapsto}}

\newcommand{\gammas}[1]{G\of{#1}}
\newcommand{\pis}[1]{P\of{#1}}

\newcommand{\clausejoin}[1]{\kappa\of{#1}} % of project-join tree node
\newcommand{\clausejoinexpand}[1]{\prod_{c \in \gammas{#1}} \pb{c}}

\newcommand{\stack}{\sigma}

\DeclareMathOperator*{\dsgn}{dsgn}
\newcommand{\gx}{d} % derivative sign

\DeclareMathOperator*{\argmax}{argmax}

\newcommand{\rand}{\reflectbox{$\mathsf{R}$}} % random quantifier

\usepackage{scalefnt} % makes big operators
\newcommand{\scaleexists}[1]{\vcenter{\hbox{\scalefont{#1}$\exists$}}}
\DeclareMathOperator*{\bigexists}{\vphantom\sum\mathchoice{\scaleexists{2}}{\scaleexists{1.4}}{\scaleexists{1}}{\scaleexists{0.75}}}
\newcommand{\scalerand}[1]{\vcenter{\hbox{\scalefont{#1}$\rand$}}}
\DeclareMathOperator*{\bigrand}{\vphantom\sum\mathchoice{\scalerand{2}}{\scalerand{1.4}}{\scalerand{1}}{\scalerand{0.75}}}

\newcommand{\score}[2]{\operatorname{score}\of{#1, #2}}

\newcommand{\eg}{e.g.}
\newcommand{\ie}{i.e.}
\newcommand{\wrt}{w.r.t.}

\newcommand{\bmpe}{Boolean MPE}

\begin{document}

%%%%%%%%%%%%%%%%%%%%%%%%%%%%%%%%%%%%%%%%%%%%%%%%%%%%%%%%%%%%%%%%%%%%%%%%%%%%%%%%

\maketitle

% \the\textwidth % 491.43787pt

\begin{abstract}
    In Bayesian inference, the \emph{maximum a posteriori (MAP)} problem combines the \emph{most probable explanation (MPE)} and \emph{marginalization (MAR)} problems.
    The counterpart in propositional logic is the \emph{exist-random stochastic satisfiability (ER-SSAT)} problem, which combines the \emph{satisfiability (SAT)} and \emph{weighted model counting (WMC)} problems.
    Both MAP and ER-SSAT have the form $\argmax_X \sum_Y f\of{X, Y}$, where $f$ is a real-valued function over disjoint sets $X$ and $Y$ of variables.
    These two optimization problems request a value assignment for the $X$ variables that maximizes the weighted sum of $f\of{X, Y}$ over all value assignments for the $Y$ variables.
    ER-SSAT has been shown to be a promising approach to formally verify fairness in supervised learning.
    Recently, \emph{dynamic programming} on graded project-join trees has been proposed to solve \emph{weighted projected model counting (WPMC)}, a related problem that has the form $\sum_X \max_Y f\of{X, Y}$.
    We extend this WPMC framework to exactly solve ER-SSAT and implement a dynamic-programming solver named \dper.
    Our empirical evaluation indicates that \dper{} contributes to the portfolio of state-of-the-art ER-SSAT solvers (\dcssat{} and \erssat) through competitive performance on low-width problem instances.
\end{abstract}

%%%%%%%%%%%%%%%%%%%%%%%%%%%%%%%%%%%%%%%%%%%%%%%%%%%%%%%%%%%%%%%%%%%%%%%%%%%%%%%%

\section{Introduction}

Bayesian networks \cite{pearl1985bayesian} are used in many application domains, such as medicine \cite{shwe1991probabilistic} and industrial systems \cite{cai2017bayesian}.
In Bayesian inference, the well-known \emph{most probable explanation (MPE)} problem requests a variable instantiation with the highest probability \cite{park2002using}.
MPE has the form $\argmax_X f\of{X}$, where $f$ is a real-valued function over a set $X$ of variables.
Also having that form is the \emph{satisfiability (SAT)} problem, which requests a \emph{satisfying truth assignment (model)} of a Boolean formula \cite{cook1971complexity}.
SAT can be viewed as a counterpart of MPE where $f$ is a Boolean function \cite{littman2001stochastic}.

Another fundamental problem in Bayesian inference is \emph{marginalization (MAR)}, which requests the marginal probability of observing a given piece of evidence \cite{chavira2005exploiting}.
MAR has the form $\sum_X f\of{X}$, where the summation is over variable instantiations that are consistent with the evidence.
Also having that form is the \emph{weighted model counting (WMC)} problem, which requests the sum of weights of models of a Boolean formula \cite{valiant1979complexity}.
Indeed, MAR can be solved via reduction to WMC \cite{sang2005performing}.

A generalization of MPE and MAR is the \emph{maximum a posteriori (MAP)} problem, which requests a value assignment for some variables that results in the highest marginal probability (over the remaining variables) in a Bayesian network \cite{park2004complexity}.
MAP has the form $\argmax_X \sum_Y f\of{X, Y}$, where $X$ and $Y$ are disjoint sets of variables.
Also having that form is the \emph{exist-random stochastic satisfiability (ER-SSAT)} problem.
Given a Boolean formula $\phi$ over $X \cup Y$, ER-SSAT requests a truth assignment $\ta_X$ for $X$ that maximizes the sum of weights of models $\ta_Y$ of the residual formula $\restrict{\phi}{\ta_X}$ over $Y$ \cite{lee2018solving}.
ER-SSAT can be viewed as a generalization of SAT and WMC.

Another generalization of SAT and WMC is the \emph{weighted projected model counting (WPMC)} problem, which has the form $\sum_X \max_Y f\of{X, Y}$.
Given a Boolean formula $\phi$ over $X \cup Y$, WPMC requests the sum of weights of truth assignments $\ta_X$ for $X$ such that there exists a truth assignment $\ta_Y$ to $Y$ where $\ta_X \cup \ta_Y$ is a model of $\phi$ \cite{zawadzki2013generalization}.
A recent WPMC framework, \procount{} \cite{dudek2021procount}, employs \emph{dynamic programming} \cite{bellman1966dynamic}, a well-known strategy to solve a large instance by decomposing it into parts then combining partial solutions into the final answer.
Dynamic programming has also been applied to other problems, such as SAT \cite{pan2005symbolic}, WMC \cite{dudek2020dpmc,fichte2020exploiting}, and \emph{quantified Boolean formula (QBF)} evaluation \cite{charwat2016bdd}.

To guide dynamic programming for WPMC, \procount{} uses \emph{project-join trees} as execution plans.
Given a Boolean formula in \emph{conjunctive normal form (CNF)}, \ie, as a set of clauses, a project-join tree specifies how to conjoin clauses and project out variables.
The \emph{width} of a project-join tree is an indicator of the hardness of the instance \cite{dudek2020dpmc}.
As WPMC handles two disjoint sets $X$ and $Y$ of variables differently, a project-join tree needs to be \emph{graded}, \ie, $X$ nodes appear above $Y$ nodes.
Gradedness guarantees correctness because summation (over $X$) and maximization (over $Y$) do not commute in general \cite{dudek2021procount}.
\procount{} operates in two phases.
First, the \emph{planning phase} builds a graded project-join tree with graph-decomposition techniques \cite{robertson1991graph,strasser2017computing}.
Second, the \emph{execution phase} uses the built tree to compute the final answer, where intermediate results are represented by \emph{algebraic decision diagrams (ADDs)} \cite{bahar1997algebraic,somenzi2015cudd}.

Here, we build on the approach of \procount{} to develop an exact solver for ER-SSAT.
To find a \emph{maximizing truth assignment (maximizer)} for ER-SSAT, we leverage an iterative technique that was recently proposed to solve \emph{maximum satisfiability (\ms)} \cite{kyrillidis2022dpms} and \emph{\bmpe} \cite{phan2022dpo}.
\ms{} is an optimization version of SAT and requests a truth assignment that maximizes the number of satisfied clauses of an unsatisfiable CNF formula \cite{krentel1988complexity}.
By building on techniques from WPMC \cite{dudek2021procount}, \bmpe{} \cite{phan2022dpo}, and \ms{} \cite{kyrillidis2022dpms}, we construct \dper, a dynamic-programming ER-SSAT framework.
We then compare \dper{} to state-of-the-art solvers (\dcssat{} \cite{majercik2005dc} and \erssat{} \cite{lee2018solving}).
We observe that \dper{} contributes to the portfolio of exact ER-SSAT solvers, through very competitive performance on instances with low-width project-join trees.

%%%%%%%%%%%%%%%%%%%%%%%%%%%%%%%%%%%%%%%%%%%%%%%%%%%%%%%%%%%%%%%%%%%%%%%%%%%%%%%%

\section{Related Work}

More general than ER-SSAT is the \emph{stochastic satisfiability (SSAT)} problem, which allows arbitrary alternations of the form $\argmax_X \sum_Y \argmax_{X'} \sum_{Y'} \ldots f\of{X, Y, X', Y', \ldots}$.
A complete SSAT solver is \dcssat{} \cite{majercik2005dc}, which employs a \emph{divide-and-conquer} strategy.
\dcssat{} splits a CNF formula into subformulas, caches \emph{viable partial assignments (VPAs)} derived from solving subformulas, then merges VPAs into a solution to the original formula.
To reduce memory usage, \dcssat{} only caches VPAs that are necessary to construct the final answer.

The ER-SSAT solver \erssat{} \cite{lee2018solving} uses \emph{clause-containment learning}, a technique inspired by \emph{clause selection} \cite{janota2015solving} in QBF evaluation.
Clause-containment learning prunes the search space with blocking clauses that are deduced after trying some truth assignments.
These learned clauses are further strengthened by \emph{minimal clause selection}, \emph{clause subsumption}, and \emph{partial assignment pruning}.
The \erssat{} solver extends a WPMC tool \cite{lee2017solving} that combines SAT and WMC techniques.

While \dcssat{} and \erssat{} are exact solvers for ER-SSAT, \maxcount{} \cite{fremont2017maximum} is a probabilistically approximately correct tool for a related problem, \emph{maximum model counting}, which has the form $\max_X \sum_Y \max_Z f\of{X, Y, Z}$.
The approximation tolerance and confidence level are provided as part of the input.
To provide these guarantees, \maxcount{} relies on techniques in WPMC and sampling \cite{chakraborty2016algorithmic}.

%%%%%%%%%%%%%%%%%%%%%%%%%%%%%%%%%%%%%%%%%%%%%%%%%%%%%%%%%%%%%%%%%%%%%%%%%%%%%%%%

\section{Preliminaries}

%%%%%%%%%%%%%%%%%%%%%%%%%%%%%%%%%%%%%%%%%%%%%%%%%%%%%%%%%%%%%%%%%%%%%%%%%%%%%%%%
\subsection{Graphs}

In a \emph{graph} $G$, let $\V{G}$ denote the set of vertices.
A \emph{tree} is an undirected graph that is connected and acyclic.
We refer to a vertex of a tree as a \emph{node}.
A \emph{rooted tree} is a tree $T$ together with a distinguished node $r \in \V{T}$ called the \emph{root}.
In a rooted tree $\tup{T, r}$, each node $v \in \V{T}$ has a (possibly empty) set of \emph{children}, denoted by $\C{v}$, which contains every node $v'$ adjacent to $v$ such that the path from $v'$ to $r$ passes through $v$.
A \emph{leaf} of a rooted tree $\tup{T, r}$ is a non-root node of degree one.
Let $\Lv$ denote the set of leaves of $\tup{T, r}$.
An \emph{internal node} is a member of $\V{T} \setminus \Lv$, including the root.

%%%%%%%%%%%%%%%%%%%%%%%%%%%%%%%%%%%%%%%%%%%%%%%%%%%%%%%%%%%%%%%%%%%%%%%%%%%%%%%%
\subsection{Pseudo-Boolean Functions}

Let $B^A$ denote the set of all functions $f$ with domain $\dom{f} := A$ and codomain $B$.
The \emph{restriction} of $f$ to a set $S$ is a function defined by $\restrict{f}{S} := \set{\tup{a, b} \in f \mid a \in S}$.

From now on, every variable is binary unless noted otherwise.
A \emph{truth assignment} for a set $X$ of variables is a function $\ta : X \to \B$.

A \emph{pseudo-Boolean (PB) function} over a set $X$ of variables is a function $f : \ps{X} \to \R$.
Define $\vars{f} := X$.
We say that $f$ is \emph{constant} if $X = \emptyset$.
Given some $S \subseteq X$, a truth assignment $\ta : S \to \B$ is \emph{total} \wrt{} $f$ if $S = X$ and is \emph{partial} otherwise.
A \emph{Boolean function} is a special PB function $f : \ps{X} \to \B$.

\begin{definition}[Join]
\label{defJoin}
    Let $f : \ps{X} \to \R$ and $g : \ps{Y} \to \R$ be PB functions.
    The \emph{(multiplicative) join} of $f$ and $g$ is a PB function, denoted by $f \cdot g : \ps{X \cup Y} \to \R$, defined for each $\ta \in \ps{X \cup Y}$ by $(f \cdot g)(\ta) := f(\restrict{\ta}{X}) \cdot g(\restrict{\ta}{Y})$.
\end{definition}

Join is commutative and associative: we have $f \cdot g = g \cdot f$ as well as $(f \cdot g) \cdot h = f \cdot (g \cdot h)$ for all PB functions $f$, $g$, and $h$.
Then define $\prod_{i = 1}^n f_i := f_1 \cdot f_2 \cdot \ldots \cdot f_n$.

\begin{definition}[Existential Projection]
\label{defExistProj}
    Let $f : \ps{X} \to \R$ be a PB function and $x$ be a variable.
    The \emph{existential projection} of $f$ \wrt{} $x$ is a PB function, denoted by $\exists_x f : \ps{X \setminus \set{x}} \to \R$, defined for each $\ta \in \ps{X \setminus \set{x}}$ by $(\exists_x f)(\ta) := \max\of{f(\extend{\ta}{x}{0}), f(\extend{\ta}{x}{1})}$.
\end{definition}

Existential projection is commutative: $\exists_x (\exists_y f) = \exists_y (\exists_x f)$ for all variables $x$ and $y$.
Then define $\exists_S f := \exists_x \exists_y \ldots f$, where $S = \set{x, y, \ldots}$ is a set of variables.
By convention, $\exists_\emptyset f := f$.

\begin{definition}[Maximum]
\label{defMaximum}
    Let $f : \ps{X} \to \R$ be a PB function.
    The \emph{maximum} of $f$ is the real number $(\exists_X f)(\emptyset)$.
\end{definition}

\begin{definition}[Maximizer]
\label{defMaximizer}
    Let $f : \ps{X} \to \R$ be a PB function.
    A \emph{maximizer} of $f$ is a truth assignment $\ta \in \ps{X}$ such that $f(\ta) = (\exists_X f)(\emptyset)$.
\end{definition}

\begin{definition}[Random Projection]
\label{defRandProj}
    Let $f : \ps{X} \to \R$ be a PB function, $x$ be a variable, and $0 \le p \le 1$ be a probability.
    The \emph{random projection} of $f$ \wrt{} $x$ and $p$ is a PB function, denoted by $\rand^p_x f : \ps{X \setminus \set{x}} \to \R$, defined for each $\ta \in \ps{X \setminus \set{x}}$ by $(\rand^p_x f)(\ta) := p \cdot f(\extend{\ta}{x}{1}) + (1 - p) \cdot f(\extend{\ta}{x}{0})$.
\end{definition}
We sometimes write $\rand_x$ instead of $\rand^p_x$, leaving $p$ implicit.

We define a \emph{probability mapping} for a set $X$ of variables as a function $\pr : X \to \brackets{0, 1}$.
Note that random projection is commutative: $\rand^{p_1}_{x_1} (\rand^{p_2}_{x_2} f) = \rand^{p_2}_{x_2} (\rand^{p_1}_{x_1} f)$ for all PB functions $f$, variables $x_1, x_2$, and probabilities $p_1, p_2$.
Then given a probability mapping $\pr$ for $X$ and a set $S = \set{x_1, x_2, \ldots} \subseteq X$, define $\rand^\pr_S f := \rand^{\pr(x_1)}_{x_1} \rand^{\pr(x_2)}_{x_2} \ldots f$.
By convention, $\rand^\pr_\emptyset f := f$.
We sometimes write $\rand_S$ instead of $\rand^\pr_S$, leaving $\pr$ implicit.

%%%%%%%%%%%%%%%%%%%%%%%%%%%%%%%%%%%%%%%%%%%%%%%%%%%%%%%%%%%%%%%%%%%%%%%%%%%%%%%%
\subsection{ER-SSAT}

Given a Boolean formula $\phi$, define $\vars{\phi}$ to be the set of all variables that appear in $\phi$.
Then $\phi$ represents a Boolean function, denoted by $\pb{\phi} : \ps{\vars{\phi}} \to \B$, defined according to standard Boolean semantics.

Boolean formulas can be generalized for probabilistic domains as follows \cite{papadimitriou1985games}.
\begin{definition}[Stochastic Formula]
    A \emph{stochastic formula} is one of the following:
    \begin{enumerate}
        \item a Boolean formula,
        \item the \emph{existential quantification} of a stochastic formula $\phi$ \wrt{} a variable $x$, denoted by $\exists x \phi$, or
        \item the \emph{random quantification} of a stochastic formula $\phi$ \wrt{} a variable $x$ and a probability $p$, denoted by $\rand^p x \phi$.
    \end{enumerate}
\end{definition}
We sometimes write $\rand x$ instead of $\rand^p x$, leaving $p$ implicit.

A stochastic formula $Q_1 x_1 Q_2 x_2 \ldots \phi$ (where each $Q_i \in \set{\exists, \rand}$ and $\phi$ is a Boolean formula) represents a PB function, denoted by $\pb{Q_1 x_1 Q_2 x_2 \ldots \phi} : \ps{\vars{\phi} \setminus \set{x_1, x_2, \ldots}} \to \R$.

\begin{definition}[Stochastic Semantics]
\label{defStochasticSemantics}
    A stochastic formula is interpreted as follows:
    \begin{enumerate}
        \item A Boolean formula $\phi$ represents the Boolean function $\pb{\phi}$ according to Boolean semantics.
        \item The existential quantification $\exists x \phi$ represents a PB function, denoted by $\pb{\exists x \phi} : \ps{X \setminus \set{x}} \to \R$, defined by $\pb{\exists x \phi} := \exists_x \pb{\phi}$, where $\pb{\phi} : \ps{X} \to \R$ is the PB function represented by the stochastic formula $\phi$.
        \item The random quantification $\rand^p x \phi$ represents a PB function, denoted by $\pb{\rand^p x \phi} : \ps{X \setminus \set{x}} \to \R$, defined by $\pb{\rand^p x \phi} := \rand^p_x \pb{\phi}$, where $\pb{\phi} : \ps{X} \to \R$ is the PB function represented by the stochastic formula $\phi$.
    \end{enumerate}
\end{definition}

For brevity, define $\exists S := \exists x \exists y \ldots$ if $S = \set{x, y, \ldots}$ is a non-empty set of variables.
Also define $\rand^\pr S := \rand^{\pr(x)} x \rand^{\pr(y)} y \ldots$ if $\pr$ is a probability mapping where $S \subseteq \dom{\pr}$.
We sometimes write $\rand S$ instead of $\rand^\pr S$, leaving $\pr$ implicit.
Then a stochastic formula has the form $Q_1 X_1 Q_2 X_2 \ldots \phi$, where each $Q_i \in \set{\exists, \rand}$ with $Q_i \ne Q_{i + 1}$, each $X_i$ is a non-empty set of variables, and $\phi$ is a Boolean formula.
Several fundamental problems can be formulated with stochastic formulas.

\begin{definition}[Weighted Model Counting]
    Let $\phi$ be a Boolean formula, $X$ be $\vars{\phi}$, and $\pr$ be a probability mapping for $X$.
    The \emph{weighted model counting (WMC)} problem on $\rand^\pr X \phi$ requests the real number $\pb{\rand^\pr X \phi}(\emptyset)$.
\end{definition}

We focus on stochastic formulas with one quantifier alternation.
An \emph{exist-random (ER) formula} has the form $\exists X \rand Y \phi$, where $\phi$ is a Boolean formula.
In a similar fashion, an \emph{RE formula} has the form $\rand Y \exists X \phi$.
The two following problems concern both kinds of quantifiers.

\begin{definition}[Weighted Projected Model Counting]
    Let $\phi$ be a Boolean formula, $\set{X, Y}$ be a partition of $\vars{\phi}$, and $\pr$ be a probability mapping for $Y$.
    The \emph{weighted projected model counting (WPMC)} problem on $\rand^\pr Y \exists X \phi$ requests the real number $\pb{\rand^\pr Y \exists X \phi}(\emptyset)$.
\end{definition}

\begin{definition}[Exist-Random Stochastic Satisfiability]
\label{defErssat}
    Let $\phi$ be a Boolean formula, $\set{X, Y}$ be a partition of $\vars{\phi}$, and $\pr$ be a probability mapping for $Y$.
    The \emph{exist-random stochastic satisfiability (ER-SSAT)} problem on $\exists X \rand^\pr Y \phi$ requests the maximum and a maximizer of $\pb{\rand^\pr Y \phi}$.
\end{definition}
Note that the maximum of $\pb{\rand^\pr Y \phi}$ is the real number $\pb{\exists X \rand^\pr Y \phi}(\emptyset)$.

%%%%%%%%%%%%%%%%%%%%%%%%%%%%%%%%%%%%%%%%%%%%%%%%%%%%%%%%%%%%%%%%%%%%%%%%%%%%%%%%

\section{Solving ER-SSAT}

%%%%%%%%%%%%%%%%%%%%%%%%%%%%%%%%%%%%%%%%%%%%%%%%%%%%%%%%%%%%%%%%%%%%%%%%%%%%%%%%
\subsection{Monolithic Approach}

A Boolean formula is usually given in \emph{conjunctive normal form (CNF)}, \ie, as a set of \emph{clauses}.
A clause is a disjunction of literals, which are variables or negated variables.

Given a CNF formula $\phi$, we have the factorization $\pb{\phi} = \prod_{c \in \phi} \pb{c}$.
In this section, we present an inefficient algorithm to solve ER-SSAT that treats $\phi$ as a monolithic structure and ignores the CNF factored representation.
The next section describes a more efficient algorithm.

To find maximizers for ER-SSAT, we leverage the following idea, which originated from the \emph{basic algorithm} for PB programming \cite{crama1990basic} then was adapted for \ms{} \cite{kyrillidis2022dpms} and \bmpe{} \cite{phan2022dpo}.

\begin{definition}[Derivative Sign]
\label{defDsgn}
    Let $f : \ps{X} \to \R$ be a PB function and $x$ be a variable.
    The \emph{derivative sign} of $f$ \wrt{} $x$ is a function, denoted by $\dsgn_x f : \ps{X \setminus \set{x}} \to \B^{\set{x}}$, defined for each $\ta \in \ps{X \setminus \set{x}}$ by $\pars{\dsgn_x f}(\ta) := \set{\va{x}{1}}$ if $f(\extend{\ta}{x}{1}) \ge f(\extend{\ta}{x}{0})$, and $\pars{\dsgn_x f}(\ta) := \set{\va{x}{0}}$ otherwise.
\end{definition}
Note that the truth assignment $\ta : X \setminus \set{x} \to \B$ is partial \wrt{} $f$ (since $x$ is unassigned).
The derivative sign indicates how to extend $\ta$ into a total truth assignment $\ta' : X \to \B$ (by mapping $x$ to $0$ or $1$) in order to maximize $f(\ta')$.

The following result leads to an iterative process to find maximizers of PB functions \cite{kyrillidis2022dpms,phan2022dpo}.
\begin{restatable}[Iterative Maximization]{proposition}{rePropIterMax}
\label{propIterMax}
    Let $f : \ps{X} \to \R$ be a PB function and $x$ be a variable.
    Assume that a truth assignment $\ta$ is a maximizer of $\exists_x f : \ps{X \setminus \set{x}} \to \R$.
    Then the truth assignment $\ta \cup \pars{\dsgn_x f}(\ta)$ is a maximizer of $f$.
\end{restatable}

\cref{algoMono} can be used to find maximizers \cite{kyrillidis2022dpms,phan2022dpo}.

\begin{algorithm}[H]
\caption{Computing the maximum and a maximizer of a PB function}
\label{algoMono}
    \KwIn{$f_n$: a PB function over a set $X_n = \set{x_1, x_2, \ldots, x_n}$ of variables}
    \KwOut{$m \in \R$: the maximum of $f_n$}
    \KwOut{$\ta_n \in \ps{X_n}$: a maximizer of $f_n$}

    \DontPrintSemicolon
    \For{$i = n, n - 1, \ldots, 2, 1$}{
        $f_{i - 1} \gets \exists_{x_i} f_i$ \tcp{$f_i$ is a PB function over $X_i = \set{x_1, x_2, \ldots, x_i}$}
    }
    $m \gets f_0(\emptyset)$ \tcp{$f_0 = \exists_{X_n} f_n$ is a constant PB function over $X_0 = \emptyset$} \label{lineMonoMaximum}
    $\ta_0 \gets \emptyset$ \tcp{$\ta_0$ is a maximizer of $f_0$ (and also the only input to $f_0$)} \label{lineMonoTa0}
    \For{$i = 1, 2, \ldots, n - 1, n$}{
        $\ta_i \gets \ta_{i - 1} \cup \pars{\dsgn_{x_i} f_i}(\ta_{i - 1})$ \tcp{$\ta_i$ is a maximizer of $f_i$ since $\ta_{i - 1}$ is a maximizer of $f_{i - 1}$ (by \cref{propIterMax})} \label{lineMonoTaN}
    }
    \Return $\tup{m, \ta_n}$
\end{algorithm}

\begin{restatable}[Correctness of \cref{algoMono}]{proposition}{rePropMono}
\label{propMonoAlgo}
    Let $\phi$ be a Boolean formula, $\set{X, Y}$ be a partition of $\vars{\phi}$, and $\pr$ be a probability mapping for $Y$.
    \cref{algoMono} solves ER-SSAT on $\exists X \rand^\pr Y \phi$ given the input $f_n = \pb{\rand^\pr Y \phi}$, which is a PB function over a set $X_n = X$ of variables.
\end{restatable}

ER-SSAT on $\exists X \rand Y \phi$ can be solved by calling \cref{algoMono} with $\pb{\rand Y \phi}$ as an input.
But the PB function $\pb{\rand Y \phi}$ may be too large to fit in main memory, making the computation slow or even impossible.
In the next section, we exploit the CNF factorization of $\phi$ and propose a more efficient solution.

%%%%%%%%%%%%%%%%%%%%%%%%%%%%%%%%%%%%%%%%%%%%%%%%%%%%%%%%%%%%%%%%%%%%%%%%%%%%%%%%
\subsection{Dynamic Programming}

ER-SSAT on $\exists X \rand Y \phi$ involves the PB function $\pb{\exists X \rand Y \phi} = \exists_X \rand_Y \pb{\phi} = \exists_X \rand_Y \prod_{c \in \phi} \pb{c}$, where the CNF formula $\phi$ is a set of clauses $c$.
Instead of projecting all variables in $X$ and $Y$ after joining all clauses, we can be more efficient and project some variables early as follows \cite{dudek2021procount}.

\begin{restatable}[Early Projection]{proposition}{rePropEarlyProj}
\label{propEarlyProj}
    Let $f : \ps{X} \to \R$ and $g : \ps{Y} \to \R$ be PB functions.
    Assume that $\pr$ is a probability mapping for $X \cup Y$.
    Then for all $S \subseteq X \setminus Y$, we have $\exists_S (f \cdot g) = (\exists_S f) \cdot g$ and $\rand^\pr_S (f \cdot g) = (\rand^\pr_S f) \cdot g$.
\end{restatable}

Early projection can lead to smaller intermediate PB functions.
For example, the bottleneck in computing $\exists_S (f \cdot g)$ is $f \cdot g$ with size $s = \size{\vars{f \cdot g}} = \size{X \cup Y}$.
The bottleneck in computing $(\exists_S f) \cdot g$ is $f$ with size $s_1 = \size{\vars{f}} = \size{X}$ or is $(\exists_S f) \cdot g$ with size $s_2 = \size{\vars{(\exists_S f) \cdot g}} = \size{X \cup Y \setminus S}$.
Notice that $s \ge \max(s_1, s_2)$.
The difference is consequential since an operation on a PB function $h$ may take $\bigo{2^{\size{\vars{h}}}}$ time and space.

When there is only one kind of projection (\eg, just $\rand$ and no $\exists$), we can apply early projection systematically as in the following framework, \dpmc, which uses dynamic programming for WMC \cite{dudek2020dpmc}.

\clearpage
\begin{definition}[Project-Join Tree]
\label{defPjt}
    Let $\phi$ be a CNF formula (\ie, a set of clauses).
    A \emph{project-join tree} for $\phi$ is a tuple $\T = \tup{T, r, \gamma, \pi}$, where:
    \begin{itemize}
        \item $\tup{T, r}$ is a rooted tree,
        \item $\gamma : \Lv \to \phi$ is a bijection, and
        \item $\pi : \V{T} \setminus \Lv \to \ps{\vars{\phi}}$ is a function.
    \end{itemize}
    A project-join tree must satisfy the following criteria:
    \begin{enumerate}
        \item The set $\set{\pi(v) \mid v \in \V{T} \setminus \Lv}$ is a partition of $\vars{\phi}$, where some $\pi(v)$ sets may be empty.
        \item For each internal node $v$, variable $x \in \pi(v)$, and clause $c \in \phi$, if $x \in \vars{c}$, then the leaf $\gammai{c}$ is a descendant of $v$ in $\tup{T, r}$.
    \end{enumerate}
\end{definition}

For a leaf $v \in \Lv$, define $\vars{v} := \vars{\gamma(v)}$, \ie, the set of variables that appear in the clause $\gamma(v) \in \phi$.
For an internal node $v \in \V{T} \setminus \Lv$, define $\vars{v} := \pars{\bigcup_{v' \in \C{v}} \vars{v'}} \setminus \pi(v)$.

ER-SSAT has both kinds of projection in a particular order ($\exists X \rand Y \phi$).
Since $\exists_x \rand_y f \ne \rand_y \exists_x f$ in general, early projection must be applied carefully.
We use the following idea from \procount{} \cite{dudek2021procount}, an extension of \dpmc{} for WPMC ($\rand X \exists Y \phi$).

\begin{definition}[Graded Project-Join Tree]
\label{defGradedness}
    Let $\phi$ be a CNF formula, $\T = \tup{T, r, \gamma, \pi}$ be a project-join tree for $\phi$, and $\set{X, Y}$ be a partition of $\vars{\phi}$.
    We say that $\T$ is \emph{$\tup{X, Y}$-graded} if there are sets $\I_X$ and $\I_Y$, called \emph{grades}, that satisfy the following properties:
    \begin{enumerate}
        \item The set $\set{\I_X, \I_Y}$ is a partition of $\V{T} \setminus \Lv$.
        \item If $v \in \I_X$, then $\pi(v) \subseteq X$.
        \item If $v \in \I_Y$, then $\pi(v) \subseteq Y$.
        \item If $v_X \in \I_X$ and $v_Y \in \I_Y$, then $v_X$ is not a descendant of $v_Y$ in the rooted tree $\tup{T, r}$.
    \end{enumerate}
\end{definition}

\cref{figPjt} illustrates a graded project-join tree.

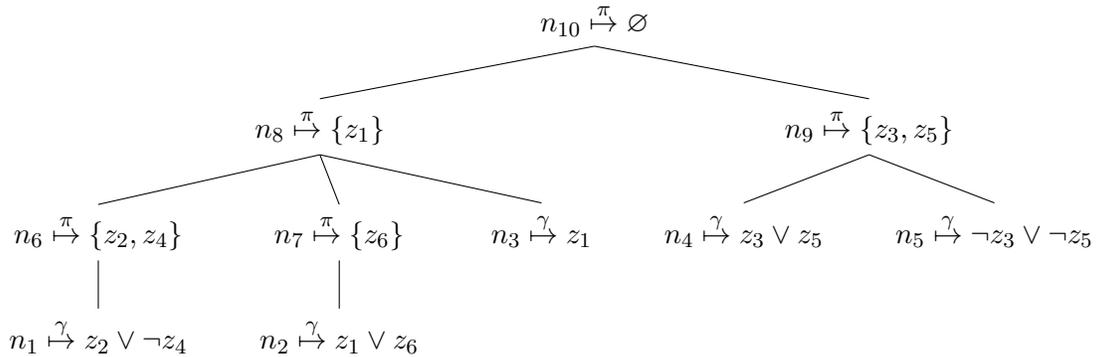
\begin{figure}[H]
    \centering
    \begin{tikzpicture}[grow = down]
        \tikzset{level distance = 40pt, sibling distance = 20pt}
        \Tree [ .$n_{10}\pimap\emptyset$
            [ .$n_8\pimap\set{z_1}$
                [ .$n_6\pimap\set{z_2, z_4}$
                    [ .$n_1\gammamap{z_2 \vee \neg z_4}$ ]
                ]
                [ .$n_7\pimap\set{z_6}$
                    [ .$n_2\gammamap{z_1 \vee z_6}$ ]
                ]
                [ .$n_3\gammamap{z_1}$ ]
            ]
            [ .$n_9\pimap\set{z_3, z_5}$
                [ .$n_4\gammamap{z_3 \vee z_5}$ ]
                [ .$n_5\gammamap{\neg z_3 \vee \neg z_5}$ ]
            ]
        ]
    \end{tikzpicture}
    \caption{
        A graded project-join tree $\T = \tup{T, n_{10}, \gamma, \pi}$ for a CNF formula $\phi$ over sets $X = \set{z_1, z_3, z_5}$ and $Y = \set{z_2, z_4, z_6}$ of variables.
        Each leaf corresponds to a clause of $\phi$ under $\gamma$.
        Each internal node corresponds to a set of variables of $\phi$ under $\pi$.
        Note that $\T$ is graded with grades $\I_X = \set{n_8, n_9, n_{10}}$ and $\I_Y = \set{n_6, n_7}$.
    }
\label{figPjt}
\end{figure}

To adapt graded project-join trees for ER-SSAT, we present the following definition, which swaps $\I_X$ and $\I_Y$ in comparison to the original version in \procount{} \cite[Definition 7]{dudek2021procount}.
\begin{definition}[Valuation of Project-Join Tree Node]
\label{defValuation}
    Let $\phi$ be a CNF formula, $\set{X, Y}$ be a partition of $\vars{\phi}$, and $\tup{T, r, \gamma, \pi}$ be an $\tup{X, Y}$-graded project-join tree for $\phi$ with grades $\I_X, \I_Y \subseteq \V{T}$.
    Also, let $\pr$ be a probability mapping for $Y$ and $v \in \V{T}$ be a node.
    The \emph{$\pr$-valuation} of $v$ is a PB function, denoted by $\val{v} : \ps{\vars{v}} \to \R$, defined by the following:
    \begin{align*}
        \val{v} :=
        \begin{cases}
            \pb{\gamma(v)} & \text{if } v \in \Lv \\
            \displaystyle
            \bigexists_{\pi(v)} \pars{ \prod_{v' \in \C{v}} \val{v'} } & \text{if } v \in \I_X \\
            \displaystyle
            \bigrand^\pr_{\pi(v)} \pars{ \prod_{v' \in \C{v}} \val{v'} } & \text{if } v \in \I_Y
        \end{cases}
    \end{align*}
\end{definition}
Recall that $\pb{\gamma(v)}$ is the Boolean function represented by the clause $\gamma(v) \in \phi$.
Also, $\val{v'}$ is the $\pr$-valuation of a child $v'$ of the node $v$ in the rooted tree $\tup{T, r}$.
Note that the $\pr$-valuation of the root $r$ is a constant PB function.

\begin{definition}[Width of Project-Join Tree]
    Let $\T = \tup{T, r, \gamma, \pi}$ be a project-join tree.
    For a leaf $v$ of $\T$, define $\width{v} := \size{\vars{v}}$.
    For an internal node $v$ of $\T$, define $\width{v} := \size{\vars{v} \cup \pi(v)}$.
    The \emph{width} of $\T$ is $\width{\T} := \max_{v \in \V{T}} \width{v}$.
\end{definition}
Note that $\width{\T}$ is the maximum number of variables needed to valuate a node $v \in \V{T}$.
Valuating $\T$ may take $\bigo{2^{\width{\T}}}$ time and space.

In comparison to the original correctness result for WPMC in \procount{} \cite[Theorem 2]{dudek2021procount}, the following version uses $\exists X \rand Y$ instead of $\rand X \exists Y$, in accordance with our \cref{defValuation}.
\begin{restatable}[Valuation of Project-Join Tree Root]{theorem}{reThmRootVal}
\label{thmRootValuation}
    Let $\phi$ be a CNF formula, $\set{X, Y}$ be a partition of $\vars{\phi}$, $\tup{T, r, \gamma, \pi}$ be an $\tup{X, Y}$-graded project-join tree for $\phi$, and $\pr$ be a probability mapping for $Y$.
    Then $\val{r}(\emptyset) = \pb{\exists X \rand^\pr Y \phi}(\emptyset)$.
\end{restatable}
In other words, the valuation of the root $r$ is a constant PB function that maps $\emptyset$ to the maximum of $\pb{\rand^\pr Y \phi}$.

We now introduce \cref{algoDp}, which is more efficient than \cref{algoMono} due to the use of a graded project-join tree to systematically apply early projection.

\clearpage

\begin{algorithm}[H]
\caption{Dynamic Programming for ER-SSAT on $\exists X \rand^\pr Y \phi$}
\label{algoDp}
    \KwIn{$\phi$: a CNF formula, where $\set{X, Y}$ is a partition of $\vars{\phi}$}
    \KwIn{$\pr$: a probability mapping for $Y$}
    \KwOut{$m \in \R$: the maximum of $\pb{\rand^\pr Y \phi}$}
    \KwOut{$\ta \in \ps{X}$: a maximizer of $\pb{\rand^\pr Y \phi}$}

    \DontPrintSemicolon
    $\T = \tup{T, r, \gamma, \pi} \gets$ an $\tup{X, Y}$-graded project-join tree for $\phi$\;
    $\stack \gets \tup{}$ \tcp{an initially empty stack}
    $\val{r} \gets \valuator(\phi, \T, \pr, r, \stack)$ \tcp{$\valuator$ (\cref{valuatorAlgo}) pushes derivative signs onto $\stack$}
    $m \gets \val{r}(\emptyset)$ \tcp{$\val{r}$ is a constant PB function}
    $\ta \gets \emptyset$ \tcp{an initially empty truth assignment}
    \While{$\stack$ is not empty}{
        $\gx \gets \pop{\stack}$ \tcp{$\gx = \dsgn_x f$ is a derivative sign, where $x \in X$ is an unassigned variable and $f$ is a PB function}
        $\ta \gets \ta \cup \gx\of{\restrict{\ta}{\dom{\gx}}}$ \tcp{$\gx\of{\restrict{\ta}{\dom{\gx}}} = \set{\va{x}{b}}$, where $b \in \B$}
    }
    \Return $\tup{m, \ta}$ \tcp{all $X$ variables have been assigned in $\ta$}
\end{algorithm}

\begin{algorithm}[H]
\caption{$\valuator(\phi, \T, \pr, v, \stack)$}
\label{valuatorAlgo}
    \KwIn{$\phi$: a CNF formula, where $\set{X, Y}$ is a partition of $\vars{\phi}$}
    \KwIn{$\T = \tup{T, r, \gamma, \pi}$: an $\tup{X, Y}$-graded project-join tree for $\phi$}
    \KwIn{$\pr$: a probability mapping for $Y$}
    \KwIn{$v \in \V{T}$: a node}
    \KwIn{$\stack$: a stack (of derivative signs) that will be modified}
    \KwOut{$\val{v}$: the $\pr$-valuation of $v$}

    \DontPrintSemicolon
    \If{$v \in \Lv$}{
        \Return $\pb{\gamma(v)}$ \tcp{the Boolean function represented by the clause $\gamma(v) \in \phi$}
    }
    \Else(\tcp*[h]{$v$ is an internal node of $\tup{T, r}$}){
        $f \gets \prod_{v' \in \C{v}} \valuator(\phi, \T, \pr, v', \stack)$\;
        \For{$x \in \pi(v)$}{
            \If(\tcp*[h]{case $\pi(v) \subseteq X$ (by gradedness)}){$x \in X$}{
                $\push{\stack}{\dsgn_x f}$ \tcp{$\stack$ is used to construct a maximizer (\cref{algoDp})}
                $f \gets \exists_x f$
            }
            \Else(\tcp*[h]{case $\pi(v) \subseteq Y$}){
                $f \gets \rand^{\pr(x)}_x f$
            }
        }
        \Return $f$
    }
\end{algorithm}

\begin{lemma}[Correctness of \cref{valuatorAlgo}]
\label{lemmaValuator}
    \cref{valuatorAlgo} returns the $\pr$-valuation of the input project-join tree node.
\end{lemma}
\begin{proof}
    \cref{valuatorAlgo} implements \cref{defValuation}.
    Modifying the input stack $\stack$ does not affect how the output valuation $f$ is computed.
\end{proof}

\begin{restatable}[Correctness of \cref{algoDp}]{theorem}{reThmDp}
\label{thmDpAlgo}
    Let $\phi$ be a CNF formula, $\set{X, Y}$ be a partition of $\vars{\phi}$, and $\pr$ be a probability mapping for $Y$.
    Then \cref{algoDp} solves ER-SSAT on $\exists X \rand^\pr Y \phi$.
\end{restatable}
\begin{proof}
    See \cref{secProofDp}.
\end{proof}

\cref{algoDp} comprises two phases: a \emph{planning phase} that builds a graded project-join tree $\T$ and an \emph{execution phase} that valuates $\T$.
The probability mapping $\pr$ is needed only in the execution phase.

We implemented \cref{algoDp} as \dper, a dynamic-programming ER-SSAT solver.
\dper{} uses a planning tool, \planner{} \cite{dudek2021procount}, that invokes a solver \cite{strasser2017computing} for \emph{tree decomposition} \cite{robertson1991graph}.
A tree decomposition of a graph $G$ is a tree $T$, where each node of $T$ corresponds to a set of vertices of $G$ (plus other technical criteria).

Also, \dper{} extends an execution tool, \executor{} \cite{dudek2021procount}, that manipulates PB functions using \emph{algebraic decision diagrams (ADDs)} \cite{bahar1997algebraic}.
An ADD is a directed acyclic graph that can compactly represent a PB function and support some polynomial-time operations.
ADDs generalize \emph{binary decision diagrams (BDDs)} \cite{bryant1986graph}, which are used to manipulate Boolean functions.

\cref{figAdd} illustrates an ADD.

\begin{figure}[H]
    \centering
    \includegraphics[
        height = 200pt,
        trim = {140pt 35pt 7pt 35pt} % <left> <lower> <right> <upper>
    ]
    {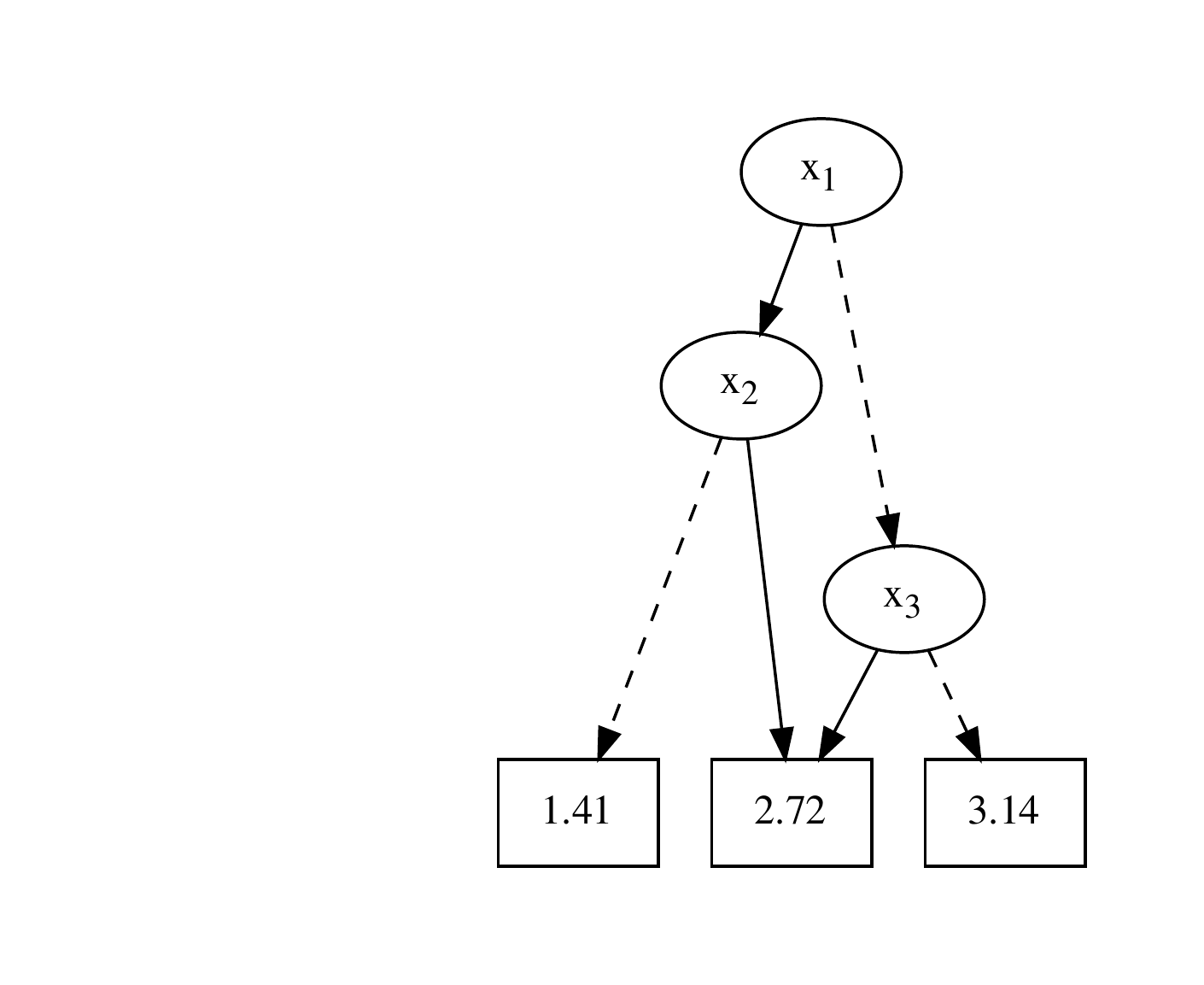}
    \caption{
        An ADD representing a PB function.
        If an edge from an oval node is solid (respectively dashed), then the corresponding variable is assigned $1$ (respectively $0$).
    }
\label{figAdd}
\end{figure}

%%%%%%%%%%%%%%%%%%%%%%%%%%%%%%%%%%%%%%%%%%%%%%%%%%%%%%%%%%%%%%%%%%%%%%%%%%%%%%%%

\section{Evaluation}

We conducted computational experiments to answer the following \emph{empirical questions}.
\begin{enumerate}
    \item Does \dper{} contribute to the portfolio of state-of-the-art exact solvers for ER-SSAT (\dcssat{} \cite{majercik2005dc} and \erssat{} \cite{lee2018solving})?
    \item Is there a class of ER formulas on which \dper{} outperforms \dcssat{} and \erssat?
\end{enumerate}

We used a high-performance computing cluster.
Each solver-benchmark pair was exclusively run on a single core of an Intel Xeon CPU (E5-2650 v2 at 2.60GHz) with a RAM cap of 100 GB and a time cap of 1000 seconds.

Code, benchmarks, and experimental data are available in a public repository:
\begin{center}
    \url{https://github.com/vuphan314/DPER}
\end{center}

%%%%%%%%%%%%%%%%%%%%%%%%%%%%%%%%%%%%%%%%%%%%%%%%%%%%%%%%%%%%%%%%%%%%%%%%%%%%%%%%
\subsection{Benchmarks}

Initially, we considered 941 stochastic formulas that were used to verify fairness in supervised learning \cite{ghosh2021justicia}.
But 729 of these benchmarks are \emph{ERE formulas} (of the form $\exists X \rand Y \exists Z \phi$), which are not the focus of this paper.
The remaining 212 benchmarks are ER formulas, but each of them was easily solved in a second.

Therefore, we considered a more challenging benchmark suite, which was used to compare WPMC solvers \cite{dudek2021procount}.
Each benchmark in the suite is an RE formula of the form $\rand^\pr Y \exists X \phi$, where $\phi$ is a CNF formula, $\set{X, Y}$ is a partition of $\vars{\phi}$, and $\pr$ be a probability mapping for $Y$.
This \cnfs-benchmark suite comprises two families.
The first family was originally used for weighted projected sampling \cite{gupta2019waps} and contains 90 RE formulas, each with a non-constant $\pr$.
The second family was originally used for (unweighted) projected sampling \cite{soos2019bird} and contains 759 RE formulas, where every probability was changed from constantly $0.5$ to randomly $0.4$ or $0.6$.
We converted these \cnfs{} harder RE formulas (of the form $\rand Y \exists X \phi$) to ER formulas (of the form $\exists X \rand Y \phi$).
Each ER-SSAT solver (\dper, \dcssat, and \erssat) was required to compute the maximum and a maximizer of $\pb{\rand Y \phi}$.

%%%%%%%%%%%%%%%%%%%%%%%%%%%%%%%%%%%%%%%%%%%%%%%%%%%%%%%%%%%%%%%%%%%%%%%%%%%%%%%%
\subsection{Correctness}

On each of the \cnfs{} ER formulas selected, if both \dcssat{} and \dper{} finished, then these solvers agreed on the maximum of $\pb{\rand Y \phi}$.
But \erssat{} returned a different number on some benchmarks.
Thus, we have high confidence that \dcssat{} and \dper{} were implemented correctly, less so about \erssat.
We notified the authors of \erssat{} about this correctness issue, which they are investigating.

To count only correct solutions while allowing some tolerance for the finite precision of floating-point arithmetic, we used a solution $m$ by \dcssat{} as the ground truth and disqualified a solution $n$ by another solver if $\abs{m - n} > 10^{-6}$.
(Note that every correct solution is a real number between $0$ and $1$.)
If $m$ was unavailable (\eg, when \dcssat{} timed out), then we simply accepted $n$.

%%%%%%%%%%%%%%%%%%%%%%%%%%%%%%%%%%%%%%%%%%%%%%%%%%%%%%%%%%%%%%%%%%%%%%%%%%%%%%%%
\subsection{Overall Performance}

We study how three solvers compare on the full set of \cnfs{} hard benchmarks.
We consider the number of solved benchmarks and another important metric, the \emph{PAR-2 score} used in SAT competitions \cite{balint2015overview,froleyks2021sat}.
Given a solver $s$ and a benchmark $b$, the PAR-2 score, $\score{s}{b}$, is the solving time if the solver solves the benchmark (within memory and time limits) and is twice the 1000-second time cap otherwise.
See \cref{tableSolvers} and \cref{figSolvers} for more detail.

\begin{table}[H]
    \centering
    \caption{
        For each solver, we compute the mean PAR-2 score and its 95\% \emph{confidence interval (CI)}.
        \vbs{0}, a \emph{virtual best solver}, simulates running \dcssat{} and \erssat{} concurrently and finishing once one of these two actual solvers succeeds.
        \vbs{1} simulates the portfolio of \dcssat, \erssat, and \dper.
        We use the differences $\delta = \score{\vbs0}{b} - \score{\vbs1}{b}$, for all \cnfs{} benchmarks $b$, to run a one-sample, right-tailed t-test with significance level $\alpha = 0.05$.
        As the p-value is $0.00036 < \alpha$, the null hypothesis ($\delta \le 0$) is rejected.
        So the improvement of \vbs1 over \vbs0 is statistically significant.
    }
    \begin{tabular}{|l|r|r|r|r|r|c|} \hline
        \multirow{2}{*}{Solver} & \multirow{2}{*}{Mean peak RAM} & \multicolumn{3}{c|}{Benchmarks solved} & \multirow{2}{*}{Mean PAR-2 score} & \multirow{2}{*}{95\% CI} \\ \cline{3-5}
        & & Unique & Fastest & Total & & \\ \hline
        \dcssat &   5.7 GB  &  82   &          221  &        286    & 1335 & $[1272, 1398]$ \\ \hline
        \erssat &   0.3 GB  & 170   &          201  &        271    & 1392 & $[1331, 1452]$ \\ \hline
        \dper   &   3.2 GB  &   9   & \fastestcnfs  &        203    & 1522 & $[1465, 1580]$ \\ \hline
        \vbs{0} &       NA  &  NA   &           NA  &        475    &  912 & $[ 846,  977]$ \\ \hline
        \vbs{1} &       NA  &  NA   &           NA  &   \vbscnfs    &  888 & $[ 822,  953]$ \\ \hline
    \end{tabular}
\label{tableSolvers}
\end{table}

\begin{figure}[H]
    \centering
    \input{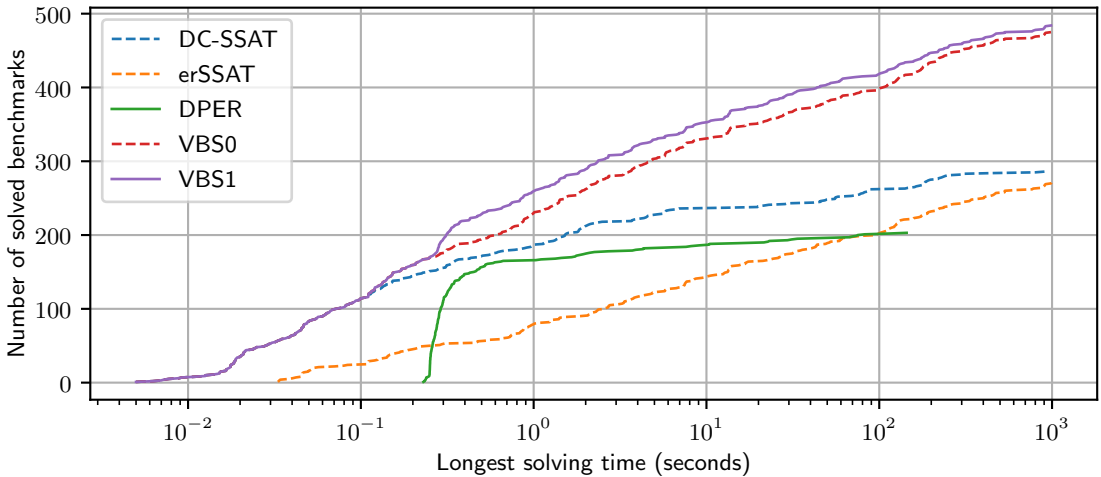}
    \caption{
        On this plot, each $\tup{x, y}$ point on a solver curve means that the solver would solve $y$ benchmarks in total if the time cap was $x$ seconds per benchmark.
        \dper{} was the fastest solver on \fastestcnfs{} (\fastestpercent) of \vbscnfs{} benchmarks that were solved by at least one tool.
        The improvement of the virtual best solver \vbs{1} over \vbs{0} is due to \dper.
    }
\label{figSolvers}
\end{figure}

Overall, \dper{} was less competitive than \dcssat{} and \erssat{} on these benchmarks.
However, to answer the first empirical question, we observe that \dper{} still has a significant contribution to the portfolio of ER-SSAT solvers by achieving the shortest solving time on \fastestcnfs{} (\fastestpercent) of \vbscnfs{} solvable benchmarks.

%%%%%%%%%%%%%%%%%%%%%%%%%%%%%%%%%%%%%%%%%%%%%%%%%%%%%%%%%%%%%%%%%%%%%%%%%%%%%%%%
\subsection{Low-Width Performance}

We focus on the subset of \plannedcnfs{} benchmarks for which graded project-join trees were successfully obtained by \dper.
The obtained widths have median $48.0$, mean $48.1$, and standard deviation $27.0$.
See \cref{tableWidths} and \cref{figWidths} for more detail.

\begin{table}[H]
    \centering
    \caption{
        We use the differences $\delta = \score{\vbs0}{b} - \score{\vbs1}{b}$, for the \plannedcnfs{} benchmarks $b$ whose project-join trees were obtained, to run a one-sample, right-tailed t-test with significance level $\alpha = 0.05$.
        As the p-value is $0.00033 < \alpha$, the null hypothesis ($\delta \le 0$) is rejected.
        So the improvement of \vbs1 over \vbs0 is statistically significant.
    }
    \begin{tabular}{|l|r|r|r|r|r|c|} \hline
        \multirow{2}{*}{Solver} & \multirow{2}{*}{Mean peak RAM} & \multicolumn{3}{c|}{Benchmarks solved} & \multirow{2}{*}{Mean PAR-2 score} & \multirow{2}{*}{95\% CI} \\ \cline{3-5}
        & & Unique & Fastest & Total & & \\ \hline
        \dcssat &   4.9 GB  &   6   &          130  &        183    &  496 & $[ 389,  603]$ \\ \hline
        \erssat &   0.1 GB  &  24   &           43  &         98    & 1205 & $[1083, 1327]$ \\ \hline
        \dper   &   2.0 GB  &   9   & \fastestcnfs  &        203    &  319 & $[ 227,  411]$ \\ \hline
        \vbs{0} &       NA  &  NA   &           NA  &        226    &  141 & $[  80,  202]$ \\ \hline
        \vbs{1} &       NA  &  NA   &           NA  &        235    &   56 & $[  17,   96]$ \\ \hline
    \end{tabular}
\label{tableWidths}
\end{table}

\begin{figure}[H]
    \centering
    \input{figures/widths.pgf}
    \caption{
        \dper{} was able to obtain graded project-join trees for \plannedcnfs{} of \cnfs{} benchmarks.
        On this plot, each $\tup{x, y}$ point on a solver curve indicates that $x$ is the central moving average of 10 consecutive project-join tree widths $1 \le w_1 < w_2 < \ldots < w_{10} \le 100$ and that $y$ is the mean PAR-2 score of the solver on benchmarks of widths $w$ such that $w_1 \le w \le w_{10}$.
        \dper{} appears to be the fastest solver on benchmarks with low widths (under $80$).
    }
\label{figWidths}
\end{figure}
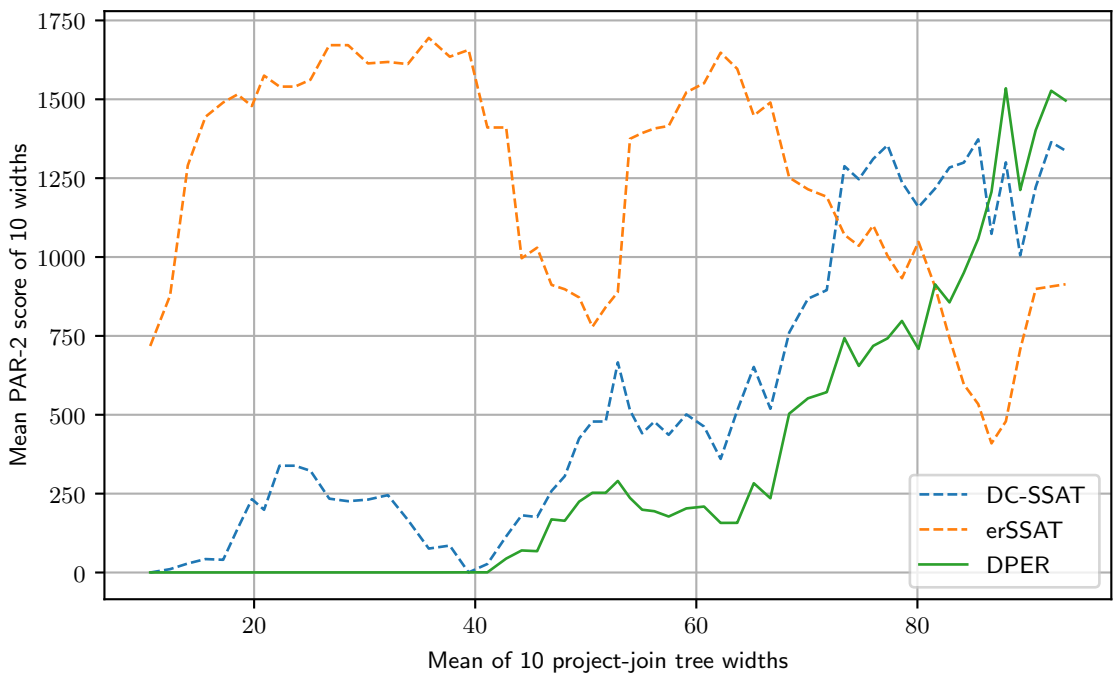

To answer the second empirical question, we identify a class of ER formulas on which \dper{} tends to be faster than \dcssat{} and \erssat.
This class comprises formulas for which there exist graded project-join trees of low widths (under $80$).

%%%%%%%%%%%%%%%%%%%%%%%%%%%%%%%%%%%%%%%%%%%%%%%%%%%%%%%%%%%%%%%%%%%%%%%%%%%%%%%%

\section{Conclusion}
\label{secConclusion}

We introduce \dper, an exact ER-SSAT solver that employs dynamic programming on graded project-join trees, based on a WPMC solver \cite{dudek2021procount}.
\dper{} also adapts an iterative procedure to find maximizers from \ms{} \cite{kyrillidis2022dpms} and \bmpe{} \cite{phan2022dpo}.
Our experiments show that \dper{} contributes to the portfolio of state-of-the-art ER-SSAT solvers (\dcssat{} \cite{majercik2005dc} and \erssat{} \cite{lee2018solving}) through competitiveness on low-width instances.

For future work, we plan to add support for hybrid inputs, such as XOR, PB, and cardinality constraints.
Also, \dper{} can be extended to solve more general problems, \eg, maximum model counting (on ERE formulas) \cite{fremont2017maximum} and \emph{functional aggregate queries} \cite{abo2016faq}.
Another research direction is to improve \dper{} with parallelism, as a portfolio solver (\eg, \cite{xu2008satzilla}) or with a multi-core ADD package (\eg, \cite{van2015sylvan}).

%%%%%%%%%%%%%%%%%%%%%%%%%%%%%%%%%%%%%%%%%%%%%%%%%%%%%%%%%%%%%%%%%%%%%%%%%%%%%%%%

% \clearpage
\appendix
\section{Proofs}

% Proofs of propositions and theorems in the main paper follow.

%%%%%%%%%%%%%%%%%%%%%%%%%%%%%%%%%%%%%%%%%%%%%%%%%%%%%%%%%%%%%%%%%%%%%%%%%%%%%%%%

\subsection{Proof of \cref{propIterMax}}

\rePropIterMax*
\begin{proof}
    By \cref{defDsgn} (derivative sign), we have $\pars{\dsgn_x f}(\ta) = \set{\va{x}{1}}$ if $f(\extend{\ta}{x}{1}) \ge f(\extend{\ta}{x}{0})$, and $\pars{\dsgn_x f}(\ta) = \set{\va{x}{0}}$ otherwise.
    First, assume the former case, $\pars{\dsgn_x f}(\ta) = \set{\va{x}{1}}$.
    Then:
    \begin{align*}
        f\of{\ta \cup \pars{\dsgn_x f}(\ta)}
        & = f\of{\extend{\ta}{x}{1}}
        \\ & = \max\of{f\of{\extend{\ta}{x}{0}}, f\of{\extend{\ta}{x}{1}}} \tag*{as we assumed the former case}
        \\ & = (\exists_x f)(\ta) \tag*{by \cref{defExistProj} (existential projection)}
        \\ & = \pars{\exists_{X \setminus \set{x}}(\exists_x f)}(\emptyset) \tag*{since $\ta$ is a maximizer of $\exists_x f : \ps{X \setminus \set{x}} \to \R$}
        \\ & = \pars{\exists_X f}(\emptyset) \tag*{because existential projection is commutative}
    \end{align*}
    So $\ta \cup \pars{\dsgn_x f}(\ta)$ is a maximizer of $f : \ps{X} \to \R$ by \cref{defMaximizer} (maximizer).
    The latter case, $\pars{\dsgn_x f}(\ta) = \set{\va{x}{0}}$, is similar.
\end{proof}

%%%%%%%%%%%%%%%%%%%%%%%%%%%%%%%%%%%%%%%%%%%%%%%%%%%%%%%%%%%%%%%%%%%%%%%%%%%%%%%%

\subsection{Proof of \cref{propMonoAlgo}}

\rePropMono*
\begin{proof}
    We prove that the two outputs of the algorithm are the maximum and a maximizer of $f_n$, as requested by the ER-SSAT problem (\cref{defErssat}).

    Regarding the {first output} ($m$), on \cref{lineMonoMaximum} of \cref{algoMono}, note that $f_0 = \exists_{x_1} \exists_{x_2} \ldots \exists_{x_{n - 1}} \exists_{x_n} f_n = \exists_{X_n} f_n$.
    Then $m = f_0(\emptyset) = (\exists_{X_n} f_n)(\emptyset)$, which is the maximum of $f_n$ by \cref{defMaximum}.

    Regarding the {second output} ($\ta_n$), Lines \ref{lineMonoTa0}-\ref{lineMonoTaN} of the algorithm iteratively compute each truth assignment $\ta_i$ that is a maximizer of the PB function $f_i$.
    This process of iterative maximization is correct due to \cref{propIterMax}.
    Finally, $\ta_n$ is a maximizer of $f_n$.
\end{proof}

%%%%%%%%%%%%%%%%%%%%%%%%%%%%%%%%%%%%%%%%%%%%%%%%%%%%%%%%%%%%%%%%%%%%%%%%%%%%%%%%

\subsection{Proof of \cref{propEarlyProj}}

\rePropEarlyProj*
\begin{proof}
    Let $x \in S$ be a variable and $p = \pr(x)$ be a probability.
    We first show that $\rand^p_x \pars{f \cdot g} = \pars{\rand^p_x f} \cdot g : \ps{X \cup Y \setminus \set{x}} \to \R$.
    For each truth assignment $\ta : X \cup Y \setminus \set{x} \to \B$, we have:
    \begin{align*}
        & \pars{\rand^p_x \pars{f \cdot g}}(\ta)
        \\ & = p \cdot \pars{f \cdot g}\of{\extend{\ta}{x}{1}} + (1 - p) \cdot \pars{f \cdot g}\of{\extend{\ta}{x}{0}} \tag*{by \cref{defRandProj} (random projection)}
        \\ & = p \cdot f\of{\restrict{\extend{\ta}{x}{1}}{X}} \cdot g\of{\restrict{\extend{\ta}{x}{1}}{Y}} + (1 - p) \cdot f\of{\restrict{\extend{\ta}{x}{0}}{X}} \cdot g\of{\restrict{\extend{\ta}{x}{0}}{Y}} \tag*{by \cref{defJoin} (join)}
        \\ & = p \cdot f\of{\restrict{\extend{\ta}{x}{1}}{X}} \cdot g\of{\restrict{\ta}{Y}} + (1 - p) \cdot f\of{\restrict{\extend{\ta}{x}{0}}{X}} \cdot g\of{\restrict{\ta}{Y}} \tag*{since $x \notin Y$}
        \\ & = \pars{p \cdot f\of{\restrict{\extend{\ta}{x}{1}}{X}} + (1 - p) \cdot f\of{\restrict{\extend{\ta}{x}{0}}{X}}} \cdot g\of{\restrict{\ta}{Y}} \tag*{by common factor}
        \\ & = \pars{p \cdot f\of{\extend{\restrict{\ta}{X}}{x}{1}} + (1 - p) \cdot f\of{\extend{\restrict{\ta}{X}}{x}{0}}} \cdot g\of{\restrict{\ta}{Y}} \tag*{as $x \in X$}
        \\ & = \pars{\rand^p_x f}\of{\restrict{\ta}{X}} \cdot g\of{\restrict{\ta}{Y}} \tag*{by definition of random projection}
        \\ & = \pars{\rand^p_x f}\of{\restrict{\ta}{X \setminus \set{x}}} \cdot g\of{\restrict{\ta}{Y}} \tag*{because $x \notin \dom{\ta} = X \cup Y \setminus \set{x}$}
        \\ & = \pars{\rand^p_x f \cdot g}(\ta) \tag*{by definition of join}
    \end{align*}
    Since random projection is commutative, we can generalize this equality from a single variable $x \in S$ to an equality on a whole set $S$ of variables: $\rand^\pr_S (f \cdot g) = (\rand^\pr_S f) \cdot g$.
    The case of existential projection, $\exists_S (f \cdot g) = (\exists_S f) \cdot g$, is similar.
\end{proof}

The following result is used later.
\begin{corollary}
\label{corEarlyProj}
    For each $i \in I = \set{1, 2, \ldots, n}$, assume that $f_i : \ps{X_i} \to \R$ is a PB function and that $S_i \subseteq X_i \setminus \bigcup_{i \ne j \in I} X_j$ is a set of variables.
    Let $\pr$ be a probability mapping such that $S = \bigcup_{i \in I} S_i \subseteq \dom{\pr}$.
    Then:
    \begin{align*}
        \prod_{i \in I} \bigrand^\pr_{S_i} f_i & = \bigrand^\pr_S \prod_{i \in I} f_i
        \\ \prod_{i \in I} \bigexists_{S_i} f_i & = \bigexists_S \prod_{i \in I} f_i
    \end{align*}
\end{corollary}
\begin{proof}
    We simply apply early projection (in the opposite direction) multiple times.
\end{proof}

%%%%%%%%%%%%%%%%%%%%%%%%%%%%%%%%%%%%%%%%%%%%%%%%%%%%%%%%%%%%%%%%%%%%%%%%%%%%%%%%

\subsection{Proof of \cref{thmRootValuation}}

To simplify notations, we use the following shorthand.

\begin{definition}[Projection Set]
\label{defProjSet}
    The \emph{projection set} of a node $v$ of a project-join tree $\tup{T, r, \pi, \gamma}$ is a set of variables, denoted by $\pis{v}$, defined by:
    \begin{align*}
        \pis{v} :=
        \begin{cases}
            \emptyset & \text{if $v \in \Lv$}
            \\ \pi(v) \cup \bigcup_{v' \in \C{v}} \pis{v'} & \text{otherwise}
        \end{cases}
    \end{align*}
\end{definition}

Also, let $\gammas{v}$ be the set of clauses $c \in \phi$ whose corresponding leaves $\gammai{c}$ are descendants of $v$ in the rooted tree $\tup{T, r}$, formally:
\begin{align*}
    \gammas{v} :=
    \begin{cases}
        \set{\gamma(v)} & \text{if $v \in \Lv$}
        \\ \bigcup_{v' \in \C{v}} \gammas{v'} & \text{otherwise}
    \end{cases}
\end{align*}
For a set $S$ of clauses $c$, define $\vars{S} := \bigcup_{c \in S} \vars{c}$.

\newcommand{\child}{u}
\newcommand{\descendant}{s}

We make the following observation.

\begin{lemma}
\label{lemmaDisjoint}
    Let $\phi$ be a CNF formula, $\T = \tup{T, r, \gamma, \pi}$ be a project-join tree for $\phi$, and $\child_1, \child_2 \in \V{T}$ be nodes.
    Assume that $\child_1$ and $\child_2$ are siblings.
    Then $\pis{\child_1} \cap \vars{\gammas{\child_2}} = \emptyset$.
\end{lemma}
\begin{proof}
    If $\pis{\child_1} = \emptyset$, then $\pis{\child_1} \cap \vars{\gammas{\child_2}} = \emptyset$.
    Now, assume that $\pis{\child_1} \ne \emptyset$.
    Let $x \in \pis{\child_1}$ be an arbitrary variable.
    Then $x \in \pi\of{\descendant}$ for some internal node $\descendant$ that is a descendant of $\child_1$.
    Let $c \in \phi$ be an arbitrary clause such that $x \in \vars{\phi}$.
    By \cref{defPjt} (project-join tree), the leaf $l = \gammai{c}$ is a descendant of $\descendant$.
    Since $\descendant$ is a descendant of $\child_1$, we know that $l$ is a descendant of $\child_1$.
    As $\child_1$ and $\child_2$ are sibling nodes (in the rooted tree $\tup{T, r}$), $l$ is not a descendant of $\child_2$.
    So $c \notin \gammas{\child_2}$.
    Thus, $x \notin \vars{\gammas{\child_2}}$ because $c$ was chosen arbitrarily such that $x \in \vars{c}$.
    Finally, $\pis{\child_1} \cap \vars{\gammas{\child_2}} = \emptyset$ because $x \in \pis{\child_1}$ was chosen arbitrary.
\end{proof}

The following shorthand is useful.
\begin{definition}[Clause Join]
\label{defClauseJoin}
    Let $v$ be a node of a project-join tree.
    The \emph{clause join} of $v$ is a Boolean function, denoted by $\clausejoin{v} : \ps{\vars{\gammas{v}}} \to \B$, defined by:
    \begin{align*}
        \clausejoin{v} := \clausejoinexpand{v}
    \end{align*}
\end{definition}

We need the following invariant regarding the valuation of a node of $\T$.

\begin{lemma}
\label{lemmaValuation}
    Let $\phi$ be a CNF formula, $\set{X, Y}$ be a partition of $\vars{\phi}$, $\T = \tup{T, r, \gamma, \pi}$ be an $\tup{X, Y}$-graded project-join tree for $\phi$, and $\pr$ be a probability mapping for $Y$.
    Then for each node $v \in \V{T}$:
    \begin{align*}
        \val{v} = \bigexists_{\pis{v} \cap X} \bigrand^\pr_{\pis{v} \cap Y} \clausejoin{v}
    \end{align*}
\end{lemma}
\begin{proof}
    Employ structural induction on $v$.

    In the \textbf{base case}, $v$ is a leaf.
    We have:
    \begin{align*}
        \val{v}
        & = \pb{\gamma(v)} \tag*{by \cref{defValuation} (valuation)}
        \\ & = \bigexists_\emptyset \bigrand^\pr_\emptyset \pb{\gamma(v)} \tag*{by convention}
        \\ & = \bigexists_{\pis{v} \cap X} \bigrand^\pr_{\pis{v} \cap Y} \pb{\gamma(v)} \tag*{as $\pis{v} = \emptyset$ for a leaf $v$}
        \\ & = \bigexists_{\pis{v} \cap X} \bigrand^\pr_{\pis{v} \cap Y} \clausejoin{v} \tag*{as $\gammas{v} = \set{\gamma(v)}$ for a leaf $v$}
    \end{align*}
    So the invariant holds for the base case.

    In the \textbf{step case}, $v$ is an internal node.
    The induction hypothesis is that
    \begin{align*}
        \val{\child} = \bigexists_{\pis{\child} \cap X} \bigrand^\pr_{\pis{\child} \cap Y} \clausejoin{\child}
    \end{align*}
    for all $\child \in \C{v}$.
    Let $S = \bigcup_{\child \in \C{v}} \pis{\child}$.
    We have the following equation chain:
    \begin{align}
        & \prod_{\child \in \C{v}} \val{\child} \label{eqChildJoin}
        \\ & = \prod_{\child \in \C{v}} \bigexists_{\pis{\child} \cap X} \bigrand^\pr_{\pis{\child} \cap Y} \clausejoin{\child} \tag*{by the induction hypothesis}
        \\ & = \bigexists_{S \cap X} \prod_{\child \in \C{v}} \bigrand^\pr_{\pis{\child} \cap Y} \clausejoin{\child} \tag*{by \cref{corEarlyProj} and \cref{lemmaDisjoint}}
        \\ & = \bigexists_{S \cap X} \bigrand^\pr_{S \cap Y} \prod_{\child \in \C{v}} \clausejoin{\child} \tag*{also by \cref{corEarlyProj} and \cref{lemmaDisjoint}}
        \\ & = \bigexists_{S \cap X} \bigrand^\pr_{S \cap Y} \clausejoin{v} \tag*{as a consequent of \cref{defClauseJoin} (clause join)}
    \end{align}
    Now, assume that $\T$ has grades $\I_X$ and $\I_Y$ (recall \cref{defGradedness}).
    By the first property of gradedness, $\set{I_X, I_Y}$ is a partition of the set of internal nodes of $\T$.

    In the \textbf{first subcase}, $v \in \I_X$.
    By gradedness, $\pi(v) \subseteq X$ (second property).
    We have:
    \begin{align*}
        \val{v}
        & = \bigexists_{\pi(v)} \prod_{\child \in \C{v}} \val{\child} \tag*{by definition of valuation}
        \\ & = \bigexists_{\pi(v)} \bigexists_{S \cap X} \bigrand^\pr_{S \cap Y} \clausejoin{v} \tag*{by Equation \eqref{eqChildJoin}}
        \\ & = \bigexists_{\pi(v) \cup S \cap X} \bigrand^\pr_{S \cap Y} \clausejoin{v} \tag*{as existential projection is commutative}
        \\ & = \bigexists_{\pis{v} \cap X} \bigrand^\pr_{S \cap Y} \clausejoin{v} \tag*{since $\pi(v) \cup S = \pis{v}$ by \cref{defProjSet} (projection set)}
        \\ & = \bigexists_{\pis{v} \cap X} \bigrand^\pr_{\pis{v} \cap Y} \clausejoin{v} \tag*{because $\pi(v) \cap Y = \emptyset$}
    \end{align*}
    So the invariant holds for the first subcase.

    In the \textbf{second subcase}, $v \in \I_Y$.
    By gradedness, $\pi(v) \subseteq Y$ (third property) and $\pis{v} \cap X = \emptyset$ (consequent of fourth property).
    We have:
    \begin{align*}
        \val{v}
        & = \bigrand^\pr_{\pi(v)} \prod_{\child \in \C{v}} \val{\child} \tag*{by definition of valuation}
        \\ & = \bigrand^\pr_{\pi(v)} \bigexists_{S \cap X} \bigrand^\pr_{S \cap Y} \clausejoin{v} \tag*{by Equation \eqref{eqChildJoin}}
        \\ & = \bigrand^\pr_{\pi(v)} \bigexists_\emptyset \bigrand^\pr_{S \cap Y} \clausejoin{v} \tag*{since $S \subseteq \pis{v}$, and $\pis{v} \cap X = \emptyset$ as mentioned before}
        \\ & = \bigrand^\pr_{\pi(v)} \bigrand^\pr_{S \cap Y} \clausejoin{v} \tag*{by convention}
        \\ & = \bigrand^\pr_{\pi(v) \cup S \cap Y} \clausejoin{v} \tag*{as random projection is commutative}
        \\ & = \bigrand^\pr_{\pis{v} \cap Y} \clausejoin{v} \tag*{by definition of projection set}
        \\ & = \bigexists_\emptyset \bigrand^\pr_{\pis{v} \cap Y} \clausejoin{v} \tag*{by convention}
        \\ & = \bigexists_{\pis{v} \cap X} \bigrand^\pr_{\pis{v} \cap Y} \clausejoin{v} \tag*{because $\pis{v} \cap X = \emptyset$ as mentioned before}
    \end{align*}
    So the invariant holds for the second subcase.
\end{proof}

We are now ready to prove \cref{thmRootValuation}.

\reThmRootVal*
\begin{proof}
    By \cref{defPjt} (project-join tree), the set $\set{\pi(v) \mid v \in \V{T} \setminus \Lv}$ is a partition of $\vars{\phi}$.
    Since $\pis{r} = \bigcup_{v \in \V{T} \setminus \Lv} \pi(v)$, we know that $\pis{r} = \vars{\phi} = X \cup Y$.
    Then:
    \begin{align*}
        \val{r}
        & = \bigexists_{\pis{r} \cap X} \bigrand^\pr_{\pis{r} \cap Y} \clausejoin{r} \tag*{by \cref{lemmaValuation}}
        \\ & = \bigexists_X \bigrand^\pr_Y \clausejoin{r} \tag*{since $\pis{r} = X \cup Y$ and $X \cap Y = \emptyset$}
        \\ & = \bigexists_X \bigrand^\pr_Y \clausejoinexpand{r} \tag*{by definition of clause join}
        \\ & = \bigexists_X \bigrand^\pr_Y \prod_{c \in \phi} \pb{c} \tag*{because $r$ is the root of a project-join tree}
        \\ & = \bigexists_X \bigrand^\pr_Y \pb{\phi} \tag*{by Boolean semantics of conjunction}
        \\ & = \bigexists_X \pb{\bigrand^\pr Y \phi} \tag*{by stochastic semantics of random quantification (\cref{defStochasticSemantics})}
        \\ & = \pb{\bigexists X \bigrand^\pr Y \phi} \tag*{by stochastic semantics of existential quantification}
    \end{align*}
    Therefore:
    \begin{align*}
        \val{r}(\emptyset) = \pb{\bigexists X \bigrand^\pr Y \phi}(\emptyset)
    \end{align*}
\end{proof}

%%%%%%%%%%%%%%%%%%%%%%%%%%%%%%%%%%%%%%%%%%%%%%%%%%%%%%%%%%%%%%%%%%%%%%%%%%%%%%%%

\subsection{Proof of \cref{thmDpAlgo}}
\label{secProofDp}

\cref{thmDpAlgo} concerns \cref{algoDp} and \cref{valuatorAlgo}.
We actually prove the correctness of their annotated versions, which are respectively \cref{dpAlgoA} and \cref{valuatorAlgoA}.

\newcommand{\actives}{A} % multiset of active PB functions
\newcommand{\elims}{E} % set of existentially eliminated variables
\newcommand{\childval}{h}

To simplify notations, for a multiset $\actives$ of PB functions $a$, define $\pb{\actives} := \prod_{a \in \actives} a$ and $\vars{\actives} := \bigcup_{a \in \actives} \vars{a}$.

\clearpage
\begin{algorithm}[H]
\caption{Dynamic Programming for ER-SSAT on $\exists X \rand^\pr Y \phi$}
\label{dpAlgoA}
    \KwIn{$\phi$: a CNF formula, where $\set{X, Y}$ is a partition of $\vars{\phi}$}
    \KwIn{$\pr$: a probability mapping for $Y$}
    \KwOut{$m \in \R$: the maximum of $\pb{\rand^\pr Y \phi}$}
    \KwOut{$\ta \in \ps{X}$: a maximizer of $\pb{\rand^\pr Y \phi}$}

    \DontPrintSemicolon
    $\T = \tup{T, r, \gamma, \pi} \gets$ an $\tup{X, Y}$-graded project-join tree for $\phi$\;
    $\stack \gets \tup{}$ \tcp{an initially empty stack}
    $\elims \gets \emptyset$ \tcp{$\elims \subseteq \vars{\phi}$ is a set of variables that have been eliminated via projection} \label{lineInitProjVarsA}
    $\actives \gets \set{\pb{c} \mid c \in \phi}$ \tcp{$\actives$ is a multiset of ``active'' PB functions} \label{lineInitActivesA}
    $\val{r} \gets \valuator(\phi, \T, \pr, r, \stack, \elims, \actives)$ \tcp{$\valuator$ (\cref{valuatorAlgoA}) modifies $\stack$, $\elims$, and $\actives$} \label{lineValRootA}
    $m \gets \val{r}(\emptyset)$ \tcp{$\val{r}$ is a constant PB function}
    $\ta \gets \emptyset$ \tcp{an initially empty truth assignment}
    \Assert{$\ta$ is a maximizer of $\exists_{\elims \cap X} \rand^\pr_{\elims \cap Y} \pb{\phi}$} \tcp{$\elims = \vars{\phi}$ after \cref{lineValRootA}} \label{lineMaximizerConstA}
    \While{$\stack$ is not empty}{
        $\gx \gets \pop{\stack}$ \tcp{$\gx = \dsgn_x f$ is a derivative sign, where $x \in \elims \cap X$ is an unassigned variable and $f$ is a PB function}
        $\ta' \gets \ta$\;
        $\ta \gets \ta' \cup \gx\of{\restrict{\ta'}{\dom{\gx}}}$ \tcp{$\gx\of{\restrict{\ta'}{\dom{\gx}}} = \set{\va{x}{b}}$, where $b \in \B$}
        $\elims \gets \elims \setminus \set{x}$ \tcp{$x$ has just been assigned $b$ in $\ta$}
        \Assert{$\ta$ is a maximizer of $\exists_{\elims \cap X} \rand^\pr_{\elims \cap Y} \pb{\phi}$}\; \label{lineMaximizerPopA}
    }
    \Return $\tup{m, \ta}$ \tcp{$\elims = Y$ now (all $X$ variables have been assigned in $\ta$)}
\end{algorithm}

\clearpage
\begin{algorithm}[H]
\caption{$\valuator(\phi, \T, \pr, v, \stack, \elims, \actives)$}
\label{valuatorAlgoA}
    \KwIn{$\phi$: a CNF formula, where $\set{X, Y}$ is a partition of $\vars{\phi}$}
    \KwIn{$\T = \tup{T, r, \gamma, \pi}$: an $\tup{X, Y}$-graded project-join tree for $\phi$}
    \KwIn{$\pr$: a probability mapping for $Y$}
    \KwIn{$v \in \V{T}$: a node}
    \KwIn{$\stack$: a stack (of derivative signs) that will be modified}
    \KwIn{$\elims$: a set (of variables) that will be modified}
    \KwIn{$\actives$: a multiset (of PB functions) that will be modified}
    \KwOut{$\val{v}$: the $\pr$-valuation of $v$}

    \DontPrintSemicolon
    \Assert{$\pb{\actives} = \exists_{\elims \cap X} \rand^\pr_{\elims \cap Y} \pb{\phi}$} \tcp{pre-condition} \label{linePrecondA}
    \If{$v \in \Lv$}{
        $f \gets \pb{\gamma(v)}$ \tcp{the Boolean function represented by the clause $\gamma(v) \in \phi$}
    }
    \Else(\tcp*[h]{$v$ is an internal node of $\tup{T, r}$}){
        $f \gets \set{\tup{\emptyset, 1}}$ \tcp{the PB function for multiplicative identity}
        $\minsert{\actives}{f}$\; \label{lineInsertIdA}
        \For(\tcp*[h]{non-empty set of nodes}){$\child \in \C{v}$}{
            $\childval \gets \valuator(\phi, \T, \pr, \child, \stack, \elims, \actives)$ \tcp{recursive call on a child $\child$ of $v$} \label{lineValChildA}
            $f' \gets f$\;
            $f \gets f' \cdot \childval$\;
            $\mremove{\actives}{\childval}$ \tcp{$\childval$ was added by initialization (\cref{lineInitActivesA} of \cref{dpAlgoA}) or a recursive call (\cref{lineInsertJoinA} or \ref{lineInsertProjA})}
            $\mremove{\actives}{f'}$ \tcp{$f'$ was added before this loop (\cref{lineInsertIdA}) or in the previous iteration (\cref{lineInsertJoinA})}
            $\minsert{\actives}{f}$\; \label{lineInsertJoinA}
        }
        \Assert{$\pb{\actives} = \exists_{\elims \cap X} \rand^\pr_{\elims \cap Y} \pb{\phi}$} \tcp{join-condition} \label{lineJoinCondA}
        \For(\tcp*[h]{possibly empty set of variables}){$x \in \pi(v)$}{
            $f' \gets f$\;
            \If(\tcp*[h]{case $\pi(v) \subseteq X$ (by gradedness)}){$x \in X$}{
                $\gx \gets \dsgn_x f$\;
                $\push{\stack}{\gx}$ \tcp{$\stack$ will be used to construct a maximizer (\cref{dpAlgoA})}
                \Assert{$\ta \cup \gx\of{\restrict{\ta}{\dom{\gx}}}$ is a maximizer of $\exists_{\elims \cap X} \rand^\pr_{\elims \cap Y} \pb{\phi}$ if $\ta$ is a maximizer of $\exists_{\set{x} \cup \elims \cap X} \rand^\pr_{\elims \cap Y} \pb{\phi}$}\; \label{lineMaximizerPushA}
                $f \gets \exists_x f'$\;
            }
            \Else(\tcp*[h]{case $\pi(v) \subseteq Y$}){
                $f \gets \rand^{\pr(x)}_x f'$
            }
            $\elims \gets \elims \cup \set{x}$\;
            $\mremove{\actives}{f'}$ \tcp{$f'$ was added before this loop (\cref{lineInsertJoinA}) or in the previous iteration (\cref{lineInsertProjA})}
            $\minsert{\actives}{f}$\; \label{lineInsertProjA}
            \Assert{$\pb{\actives} = \exists_{\elims \cap X} \rand^\pr_{\elims \cap Y} \pb{\phi}$} \tcp{project-condition} \label{lineProjCondA}
        }
    }
    \Assert{$\pb{\actives} = \exists_{\elims \cap X} \rand^\pr_{\elims \cap Y} \pb{\phi}$} \tcp{post-condition} \label{linePostCondA}
    \Return $f$
\end{algorithm}

\begin{restatable}[Correctness of \cref{dpAlgoA}]{theorem}{reAnnThmDper}
\label{thmDperA}
    Let $\phi$ be a CNF formula, $\set{X, Y}$ be a partition of $\vars{\phi}$, and $\pr$ be a probability mapping for $Y$.
    Then \cref{dpAlgoA} solves ER-SSAT on $\exists X \rand^\pr Y \phi$.
\end{restatable}

We will first prove some lemmas regarding \cref{valuatorAlgoA}.

%%%%%%%%%%%%%%%%%%%%%%%%%%%%%%%%%%%%%%%%%%%%%%%%%%%%%%%%%%%%%%%%%%%%%%%%%%%%%%%%
\subsection{Proofs for \cref{valuatorAlgoA}}

The following lemmas involve a CNF formula $\phi$, a partition $\set{X, Y}$ of $\vars{\phi}$, a probability mapping $\pr$ for $Y$, and an $\tup{X, Y}$-graded project-join tree $\T = \tup{T, r, \gamma, \pi}$ for $\phi$.

\begin{lemma}
\label{lemmaPreCondRootA}
    The pre-condition of the call $\valuator(\phi, \T, \pr, r, \stack, \elims, \actives)$ holds (\cref{linePrecondA} of \cref{valuatorAlgoA}).
\end{lemma}
\begin{proof}
    \cref{valuatorAlgoA} is called for the first time on \cref{lineValRootA} of \cref{dpAlgoA}.
    We have:
    \begin{align*}
        \pb{\actives}
        & = \prod_{a \in \actives} a \tag*{by definition}
        \\ & = \prod_{c \in \phi} \pb{c} \tag*{as initialized on \cref{lineInitActivesA} of \cref{dpAlgoA}}
        \\ & = \pb{\phi} \tag*{because $\phi$ is a CNF formula}
        \\ & = \exists_\emptyset \rand^\pr_\emptyset \pb{\phi} \tag*{by convention}
        \\ & = \exists_{\elims \cap X} \rand^\pr_{\elims \cap Y} \pb{\phi} \tag*{as initialized on \cref{lineInitProjVarsA} of \cref{dpAlgoA}}
    \end{align*}
\end{proof}

To simplify proofs, for each internal node $v$ of $\T$, assume that $\C{v}$ and $\pi(v)$ have arbitrary (but fixed) orders.
Then we can refer to members of these two sets as the first, second, \ldots, and last.

\begin{lemma}
\label{lemmaPreCondFirstChildA}
    Let $v$ be an internal node of $\T$ and $\child$ be the first node in the set $\C{v}$.
    Note that $\valuator(\phi, \T, \pr, v, \stack, \elims, \actives)$ calls $\valuator(\phi, \T, \pr, \child, \stack, \elims, \actives)$.
    Assume that the pre-condition of the caller ($v$) holds.
    Then the pre-condition of the callee ($\child$) holds.
\end{lemma}
\begin{proof}
    Before \cref{lineValChildA} of \cref{valuatorAlgoA}, $\actives$ is modified trivially: inserting the multiplicative identity $f$ does not change $\pb{\actives}$.
    Also, $\elims$ is not modified.
\end{proof}

\begin{lemma}
\label{lemmaPostCondLeafA}
    Let $v$ be a leaf of $\T$.
    Assume that the pre-condition of the call $\valuator(\phi, \T, \pr, v, \stack, \elims, \actives)$ holds.
    Then the post-condition holds (\cref{linePostCondA}).
\end{lemma}
\begin{proof}
    Neither $\actives$ nor $\elims$ is modified in the non-recursive branch of \cref{valuatorAlgoA}, \ie, when $v \in \Lv$.
\end{proof}

\begin{lemma}
\label{lemmaPreCondNextChildA}
    Let $v$ be an internal node of $\T$ and $\child_1 < \child_2$ be consecutive nodes in $\C{v}$.
    Note that $\valuator(\phi, \T, \pr, v, \stack, \elims, \actives)$ calls $\valuator(\phi, \T, \pr, \child_1, \stack, \elims, \actives)$ then calls $\valuator(\phi, \T, \pr, \child_2, \stack, \elims, \actives)$.
    Assume that the post-condition of the first callee ($\child_1$) holds.
    Then the pre-condition of the second callee ($\child_2$) holds.
\end{lemma}
\begin{proof}
    After the first callee returns and before the second callee starts, $\actives$ is trivially modified: replacing $\childval$ and $f'$ with $f = f' \cdot \childval$ does not change $\pb{\actives}$.
    Also, $\elims$ is not modified.
\end{proof}

\begin{lemma}
\label{lemmaJoinCondA}
    Let $v$ be an internal node of $\T$.
    For each $\child \in \C{v}$, note that $\valuator(\phi, \T, \pr, v, \stack, \elims, \actives)$ calls $\valuator(\phi, \T, \pr, \child, \stack, \elims, \actives)$.
    Assume that the post-condition of the last callee holds.
    Then the join-condition of the caller holds (\cref{lineJoinCondA}).
\end{lemma}
\begin{proof}
    After the last callee returns, $\actives$ is trivially modified: replacing $\childval$ and $f'$ with $f = f' \cdot \childval$ does not change $\pb{\actives}$.
    Also, $\elims$ is not modified.
\end{proof}

\begin{lemma}
\label{lemmaProjCondFirstVarA}
    Let $v$ be an internal node of $\T$ and $x$ be the first variable in $\pi(v)$.
    Assume that the join-condition of \cref{valuatorAlgoA} holds.
    Then the project-condition holds (\cref{lineProjCondA}).
\end{lemma}
\begin{proof}
    First, consider the case $x \in X$.
    On the project-condition line:
    \begin{align*}
        \pb{\actives}
        & = \pb{\set{f} \cup \actives \setminus \set{f}}
        \\ & = f \cdot \pb{\actives \setminus \set{f}}
        \\ & = \pars{\exists_x f'} \cdot \pb{\actives \setminus \set{f}}
        \\ & = \exists_x \pars{f' \cdot \pb{\actives \setminus \set{f}}} \tag*{since $x \notin \vars{\actives \setminus \set{f}}$, as $\T$ is a project-join tree}
        \\ & = \exists_x \pb{\set{f'} \cup \actives \setminus \set{f}}
        \\ & = \exists_x \exists_{\elims \cap X \setminus \set{x}} \rand^\pr_{\elims \cap Y} \pb{\phi} \tag*{because the join-condition was assumed to hold}
        \\ & = \exists_{\elims \cap X} \rand^\pr_{\elims \cap Y} \pb{\phi}
    \end{align*}

    Now, consider the case $x \in Y$.
    On the project-condition line:
    \begin{align*}
        \pb{\actives}
        & = \pb{\set{f} \cup \actives \setminus \set{f}}
        \\ & = f \cdot \pb{\actives \setminus \set{f}}
        \\ & = \pars{\rand^{\pr(x)}_x f'} \cdot \pb{\actives \setminus \set{f}}
        \\ & = \rand^{\pr(x)}_x \pars{f' \cdot \pb{\actives \setminus \set{f}}} \tag*{since $x \notin \vars{\actives \setminus \set{f}}$, as $\T$ is a project-join tree}
        \\ & = \rand^{\pr(x)}_x \pb{\set{f'} \cup \actives \setminus \set{f}}
        \\ & = \rand^{\pr(x)}_x \exists_{\elims \cap X} \rand^\pr_{\elims \cap Y \setminus \set{x}} \pb{\phi} \tag*{because the join-condition was assumed to hold}
        \\ & = \rand^{\pr(x)}_x \exists_{\elims \cap X} \rand^\pr_{\elims \cap Y \setminus \set{x}} \pb{\gammas{v} \cup \psi} \tag*{where $\psi = \phi \setminus \gammas{v}$}
        \\ & = \rand^{\pr(x)}_x \exists_{\elims \cap X} \rand^\pr_{\elims \cap Y \setminus \set{x}} \pars{\pb{\gammas{v}} \cdot \pb{\psi}}
        \\ & = \rand^{\pr(x)}_x \exists_{\elims \cap X} \rand^\pr_{S_1} \rand^\pr_{S_2} \pars{\pb{\gammas{v}} \cdot \pb{\psi}} \tag*{where $S_1 = \pars{\elims \cap Y \setminus \set{x}} \cap \pis{v}$ and $S_2 = \pars{\elims \cap Y \setminus \set{x}} \setminus S_1$}
        \\ & = \rand^{\pr(x)}_x \exists_{\elims \cap X} \rand^\pr_{S_1} \pars{\pb{\gammas{v}} \cdot \rand^\pr_{S_2} \pb{\psi}} \tag*{since $S_2 \cap \vars{\gammas{v}} = \emptyset$, as $\T$ is a project-join tree}
        \\ & = \rand^{\pr(x)}_x \exists_{\elims \cap X} \pars{\pars{\rand^\pr_{S_1} \pb{\gammas{v}}} \cdot \rand^\pr_{S_2} \pb{\psi}} \tag*{since $S_1 \cap \vars{\psi} = \emptyset$, as $\T$ is a project-join tree}
        \\ & = \rand^{\pr(x)}_x \pars{\pars{\rand^\pr_{S_1} \pb{\gammas{v}}} \cdot \exists_{\elims \cap X} \rand^\pr_{S_2} \pb{\psi}} \tag*{since $\elims \cap X \cap \vars{\gammas{v}} = \emptyset$, as $\T$ is $\tup{X, Y}$-graded}
        \\ & = \pars{\rand^{\pr(x)}_x \rand^\pr_{S_1} \pb{\gammas{v}}} \cdot \exists_{\elims \cap X} \rand^\pr_{S_2} \pb{\psi} \tag*{since $x \notin \vars{\psi}$, as $\T$ is a project-join tree}
        \\ & = \pars{\rand^\pr_{\set{x} \cup S_1} \pb{\gammas{v}}} \cdot \exists_{\elims \cap X} \rand^\pr_{S_2} \pb{\psi}
        \\ & = \exists_{\elims \cap X} \pars{\pars{\rand^\pr_{\set{x} \cup S_1} \pb{\gammas{v}}} \cdot \rand^\pr_{S_2} \pb{\psi}} \tag*{undoing early projection}
        \\ & = \exists_{\elims \cap X} \rand^\pr_{\set{x} \cup S_1} \pars{\pb{\gammas{v}} \cdot \rand^\pr_{S_2} \pb{\psi}} \tag*{undoing early projection}
        \\ & = \exists_{\elims \cap X} \rand^\pr_{\set{x} \cup S_1} \rand^\pr_{S_2} \pars{\pb{\gammas{v}} \cdot \pb{\psi}} \tag*{undoing early projection}
        \\ & = \exists_{\elims \cap X} \rand^\pr_{\elims \cap Y} \pars{\pb{\gammas{v}} \cdot \pb{\psi}} \tag*{because $S_1 \cup S_2 = \elims \cap Y \setminus \set{x}$}
        \\ & = \exists_{\elims \cap X} \rand^\pr_{\elims \cap Y} \pb{\gammas{v} \cup \psi}
        \\ & = \exists_{\elims \cap X} \rand^\pr_{\elims \cap Y} \pb{\phi}
    \end{align*}
\end{proof}

\begin{lemma}
\label{lemmaProjCondNextVarA}
    Let $v$ be an internal node of $\T$ and $x_1 < x_2$ be consecutive variables in $\pi(v)$.
    Assume that the project-condition of \cref{valuatorAlgoA} holds in the iteration for $x_1$.
    Then the project-condition also holds in the iteration for $x_2$.
\end{lemma}
\begin{proof}
    Similar to \cref{lemmaProjCondFirstVarA}.
\end{proof}

\begin{lemma}
\label{lemmaProjCondA}
    Let $v$ be an internal node of $\T$ and $x$ be a variable in $\pi(v)$.
    Assume that the join-condition of \cref{valuatorAlgoA} holds.
    Then the project-condition holds.
\end{lemma}
\begin{proof}
    Apply \cref{lemmaProjCondFirstVarA} and \cref{lemmaProjCondNextVarA}.
\end{proof}

\begin{lemma}
\label{lemmaPostCondInternalA}
    Let $v$ be an internal node of $\T$.
    For each $\child \in \C{v}$, note that $\valuator(\phi, \T, \pr, v, \stack, \elims, \actives)$ calls $\valuator(\phi, \T, \pr, \child, \stack, \elims, \actives)$.
    Assume that the post-condition of the last callee holds.
    Then the post-condition of the caller holds.
\end{lemma}
\begin{proof}
    By \cref{lemmaJoinCondA}, the join-condition holds.
    Then by \cref{lemmaProjCondA}, the project-condition holds in the iteration for the last variable in $\pi(v)$.
    So the post-condition holds.
\end{proof}

\begin{lemma}
\label{lemmaPrePostCondsA}
    The pre-condition and post-condition of \cref{valuatorAlgoA} hold for each input $v \in \V{T}$.
\end{lemma}
\begin{proof}
    Employ induction on the recursive calls to the algorithm.

    In the \textbf{base case}, consider the leftmost nodes of $\T$: $v_1, v_2, \ldots, v_k$, where $v_{i + 1}$ is the first node in $\C{v_i}$, $v_1$ is the root $r$, and $v_k$ is a leaf.
    We show that the pre-condition of the call $\valuator(\phi, \T, \pr, v_i, \stack, \elims, \actives)$ holds for each $i = 1, 2, \ldots, k$ and that the post-condition of the call $\valuator(\phi, \T, \pr, v_k, \stack, \elims, \actives)$ holds.
    The pre-condition holds for each $v_i$ due to \cref{lemmaPreCondRootA} and \cref{lemmaPreCondFirstChildA}.
    The post-condition holds for $v_k$ by \cref{lemmaPostCondLeafA}.

    In the \textbf{step case}, the induction hypothesis is that the pre-condition holds for each call started before the current call starts and that the post-condition holds for each call returned before the current call starts.
    We show that the pre-condition and post-condition hold for the current call $\valuator(\phi, \T, \pr, v, \stack, \elims, \actives)$.
    The pre-condition for $v$ holds due to \cref{lemmaPreCondFirstChildA} and \cref{lemmaPreCondNextChildA}.
    The post-condition for $v$ holds by \cref{lemmaPostCondLeafA} and \cref{lemmaPostCondInternalA}.
\end{proof}

\begin{lemma}
\label{lemmaJoinProjCondsA}
    The join-condition and project-condition of \cref{valuatorAlgoA} hold for each internal node $v$ of $\T$ and variable $x$ in $\pi(v)$.
\end{lemma}
\begin{proof}
    By \cref{lemmaPrePostCondsA}, the post-condition of the last call on \cref{lineValChildA} holds.
    Then by \cref{lemmaJoinCondA}, the join-condition holds.
    And by \cref{lemmaProjCondA}, the project-condition holds.
\end{proof}

\newcommand{\other}{g} % PB function

\begin{lemma}
\label{lemmaDsgnEarlyProjA}
    Let $f : \ps{X} \to \R$ and $\other : \ps{Y} \to \R$ be PB functions with non-negative ranges.
    Assume that $x$ is a variable in $X \setminus Y$ and that $\ta$ is a truth assignment for $Y \cup X \setminus \set{x}$.
    Then ${\pars{\dsgn_x f}\of{\restrict{\ta}{X \setminus \set{x}}}} = \pars{\dsgn_x (f \cdot \other)}(\ta)$.
\end{lemma}
\begin{proof}
    First, assume the case $\pars{\dsgn_x f}\of{\restrict{\ta}{X \setminus \set{x}}} = \set{\va{x}{1}}$.
    We have the following inequalities:
    \begin{align*}
        f\of{\extend{\restrict{\ta}{X \setminus \set{x}}}{x}{1}} & \ge f\of{\extend{\restrict{\ta}{X \setminus \set{x}}}{x}{0}} \tag*{by \cref{defDsgn} (derivative sign)}
        \\ f\of{\restrict{\extend{\ta}{x}{1}}{X}} & \ge f\of{\restrict{\extend{\ta}{x}{0}}{X}} \tag*{as $x \in X$}
        \\ f\of{\restrict{\extend{\ta}{x}{1}}{X}} \cdot \other\of{\restrict{\ta}{Y}} & \ge f\of{\restrict{\extend{\ta}{x}{0}}{X}} \cdot \other\of{\restrict{\ta}{Y}} \tag*{because the range of $\other$ is non-negative}
        \\ f\of{\restrict{\extend{\ta}{x}{1}}{X}} \cdot \other\of{\restrict{\extend{\ta}{x}{1}}{Y}} & \ge f\of{\restrict{\extend{\ta}{x}{0}}{X}} \cdot \other\of{\restrict{\extend{\ta}{x}{0}}{Y}} \tag*{since $x \notin Y$}
        \\ (f \cdot \other)\of{\extend{\ta}{x}{1}} & \ge (f \cdot \other)\of{\extend{\ta}{x}{0}} \tag*{by \cref{defJoin} (join)}
    \end{align*}
    So $\pars{\dsgn_x (f \cdot \other)}(\ta) = \set{\va{x}{1}}$ by definition.
    Thus ${\pars{\dsgn_x f}\of{\restrict{\ta}{X \setminus \set{x}}}} = \pars{\dsgn_x (f \cdot \other)}(\ta)$.
    The case $\pars{\dsgn_x f}\of{\restrict{\ta}{X \setminus \set{x}}} = \set{\va{x}{0}}$ is similar.
\end{proof}

\begin{lemma}
\label{lemmaMaximizerDsgnA}
    Let $f : \ps{X} \to \R$ and $\other : \ps{Y} \to \R$ be PB functions with non-negative ranges.
    Assume that $x$ is a variable in $X \setminus Y$ and that $\ta$ is a maximizer of $\exists_x (f \cdot \other)$.
    Then $\ta' = \ta \cup \pars{\dsgn_x f}\of{\restrict{\ta}{X \setminus \set{x}}}$ is a maximizer of $f \cdot \other$.
\end{lemma}
\begin{proof}
    By \cref{lemmaDsgnEarlyProjA}, $\ta' = \ta \cup \pars{\dsgn_x (f \cdot \other)}(\ta)$.
    Then by \cref{propIterMax} (iterative maximization), $\ta'$ is a maximizer of $f \cdot \other$.
\end{proof}

\begin{lemma}
\label{lemmaMaximizerPushA}
    Let $v$ be an internal node of $\T$ and $x$ be a variable in $\pi(v)$.
    Then the assertion on \cref{lineMaximizerPushA} of \cref{valuatorAlgoA} holds.
\end{lemma}
\begin{proof}
    On that line:
    \begin{align*}
        \exists_{\set{x} \cup \elims \cap X} \rand^\pr_{\elims \cap Y} \pb{\phi}
        & = \exists_x \exists_{\elims \cap X} \rand^\pr_{\elims \cap Y} \pb{\phi}
        \\ & = \exists_x \pb{\actives} \tag*{as the join-condition and project-condition hold by \cref{lemmaJoinProjCondsA}}
        \\ & = \exists_x \pb{\set{f} \cup \actives \setminus \set{f}}
        \\ & = \exists_x \pars{f \cdot \pb{\actives \setminus \set{f}}}
    \end{align*}
    Apply \cref{lemmaMaximizerDsgnA}.
\end{proof}

%%%%%%%%%%%%%%%%%%%%%%%%%%%%%%%%%%%%%%%%%%%%%%%%%%%%%%%%%%%%%%%%%%%%%%%%%%%%%%%%
\subsection{Proofs for \cref{dpAlgoA}}

\begin{lemma}
\label{lemmaMaximizerConstA}
    The assertion on \cref{lineMaximizerConstA} of \cref{dpAlgoA} holds.
\end{lemma}
\begin{proof}
    \cref{lineValRootA} changed $\elims$ to $\vars{\phi} = X \cup Y$ after all calls to $\valuator$ returned.
    Then $\exists_{\elims \cap X} \rand^\pr_{\elims \cap Y} \pb{\phi}$ is a constant PB function.
    So $\ta = \emptyset$ is the maximizer of $\exists_{\elims \cap X} \rand^\pr_{\elims \cap Y} \pb{\phi}$.
\end{proof}

\begin{lemma}
\label{lemmaMaximizerPopFirstDsgnA}
    Let $\gx$ be the first derivative sign to be popped in \cref{dpAlgoA}.
    Then the assertion on \cref{lineMaximizerPopA} holds at this time.
\end{lemma}
\begin{proof}
    Note that $\gx = \dsgn_x f$ for some variable $x \in X$ and PB function $f$.
    Consider the following on \cref{lineMaximizerPopA}.
    By \cref{lemmaMaximizerConstA}, the truth assignment $\ta'$ is a maximizer of $\exists_{\set{x} \cup \elims \cap X} \rand^\pr_{\elims \cap Y} \pb{\phi}$.
    Recall that when $\gx$ was pushed onto $\stack$, the assertion on \cref{lineMaximizerPushA} of \cref{valuatorAlgoA} holds (by \cref{lemmaMaximizerPushA}).
    Then $\ta$ is a maximizer of $\exists_{\elims \cap X} \rand^\pr_{\elims \cap Y} \pb{\phi}$.
\end{proof}

\begin{lemma}
\label{lemmaMaximizerPopNextDsgnA}
    Let \cref{dpAlgoA} successively pop derivative signs $\gx_1$ then $\gx_2$.
    Assume that the assertion on \cref{lineMaximizerPopA} holds in the iteration for $\gx_1$.
    Then the assertion also holds in the iteration for $\gx_2$.
\end{lemma}
\begin{proof}
    Note that $\gx_2 = \dsgn_{x_2} f_2$ for some variable $x_2 \in X$ and PB function $f_2$.
    Consider the following on \cref{lineMaximizerPopA} in the iteration for $\gx_2$.
    By the assumption, the truth assignment $\ta'$ is a maximizer of $\exists_{\set{x_2} \cup \elims \cap X} \rand^\pr_{\elims \cap Y} \pb{\phi}$.
    Recall that when $\gx_2$ was pushed onto $\stack$, the assertion on \cref{lineMaximizerPushA} of \cref{valuatorAlgoA} holds (by \cref{lemmaMaximizerPushA}).
    Then $\ta$ is a maximizer of $\exists_{\elims \cap X} \rand^\pr_{\elims \cap Y} \pb{\phi}$.
\end{proof}

We are now ready to prove \cref{thmDperA}.

\reAnnThmDper*
\begin{proof}
    The first output ($m$) is the maximum of $\pb{\rand^\pr Y \phi}$ as \cref{lineValRootA} computes the $\pr$-valuation of the root $r$ of a project-join tree for $\phi$ (see \cref{thmRootValuation} and \cref{lemmaValuator}).
    The second output ($\ta$) is a maximizer of $\pb{\rand^\pr Y \phi}$ by \cref{lemmaMaximizerPopFirstDsgnA} and \cref{lemmaMaximizerPopNextDsgnA}.
\end{proof}

%%%%%%%%%%%%%%%%%%%%%%%%%%%%%%%%%%%%%%%%%%%%%%%%%%%%%%%%%%%%%%%%%%%%%%%%%%%%%%%%

% \clearpage
\bibliography{dper.bib}

%%%%%%%%%%%%%%%%%%%%%%%%%%%%%%%%%%%%%%%%%%%%%%%%%%%%%%%%%%%%%%%%%%%%%%%%%%%%%%%%

\end{document}

%% file: figures/widths.pgf
%% Creator: Matplotlib, PGF backend
%%
%% To include the figure in your LaTeX document, write
%%   \input{<filename>.pgf}
%%
%% Make sure the required packages are loaded in your preamble
%%   \usepackage{pgf}
%%
%% Also ensure that all the required font packages are loaded; for instance,
%% the lmodern package is sometimes necessary when using math font.
%%   \usepackage{lmodern}
%%
%% Figures using additional raster images can only be included by \input if
%% they are in the same directory as the main LaTeX file. For loading figures
%% from other directories you can use the `import` package
%%   \usepackage{import}
%%
%% and then include the figures with
%%   \import{<path to file>}{<filename>.pgf}
%%
%% Matplotlib used the following preamble
%%   \usepackage{fontspec}
%%   \setmainfont{DejaVuSerif.ttf}[Path=\detokenize{/home/vhp1/.local/lib/python3.8/site-packages/matplotlib/mpl-data/fonts/ttf/}]
%%   \setsansfont{DejaVuSans.ttf}[Path=\detokenize{/home/vhp1/.local/lib/python3.8/site-packages/matplotlib/mpl-data/fonts/ttf/}]
%%   \setmonofont{DejaVuSansMono.ttf}[Path=\detokenize{/home/vhp1/.local/lib/python3.8/site-packages/matplotlib/mpl-data/fonts/ttf/}]
%%
\begingroup%
\makeatletter%
\begin{pgfpicture}%
\pgfpathrectangle{\pgfpointorigin}{\pgfqpoint{6.000712in}{3.675133in}}%
\pgfusepath{use as bounding box, clip}%
\begin{pgfscope}%
\pgfsetbuttcap%
\pgfsetmiterjoin%
\pgfsetlinewidth{0.000000pt}%
\definecolor{currentstroke}{rgb}{1.000000,1.000000,1.000000}%
\pgfsetstrokecolor{currentstroke}%
\pgfsetstrokeopacity{0.000000}%
\pgfsetdash{}{0pt}%
\pgfpathmoveto{\pgfqpoint{0.000000in}{0.000000in}}%
\pgfpathlineto{\pgfqpoint{6.000712in}{0.000000in}}%
\pgfpathlineto{\pgfqpoint{6.000712in}{3.675133in}}%
\pgfpathlineto{\pgfqpoint{0.000000in}{3.675133in}}%
\pgfpathlineto{\pgfqpoint{0.000000in}{0.000000in}}%
\pgfpathclose%
\pgfusepath{}%
\end{pgfscope}%
\begin{pgfscope}%
\pgfsetbuttcap%
\pgfsetmiterjoin%
\definecolor{currentfill}{rgb}{1.000000,1.000000,1.000000}%
\pgfsetfillcolor{currentfill}%
\pgfsetlinewidth{0.000000pt}%
\definecolor{currentstroke}{rgb}{0.000000,0.000000,0.000000}%
\pgfsetstrokecolor{currentstroke}%
\pgfsetstrokeopacity{0.000000}%
\pgfsetdash{}{0pt}%
\pgfpathmoveto{\pgfqpoint{0.630692in}{0.494721in}}%
\pgfpathlineto{\pgfqpoint{5.900712in}{0.494721in}}%
\pgfpathlineto{\pgfqpoint{5.900712in}{3.575133in}}%
\pgfpathlineto{\pgfqpoint{0.630692in}{3.575133in}}%
\pgfpathlineto{\pgfqpoint{0.630692in}{0.494721in}}%
\pgfpathclose%
\pgfusepath{fill}%
\end{pgfscope}%
\begin{pgfscope}%
\pgfpathrectangle{\pgfqpoint{0.630692in}{0.494721in}}{\pgfqpoint{5.270020in}{3.080412in}}%
\pgfusepath{clip}%
\pgfsetrectcap%
\pgfsetroundjoin%
\pgfsetlinewidth{0.803000pt}%
\definecolor{currentstroke}{rgb}{0.690196,0.690196,0.690196}%
\pgfsetstrokecolor{currentstroke}%
\pgfsetdash{}{0pt}%
\pgfpathmoveto{\pgfqpoint{1.414136in}{0.494721in}}%
\pgfpathlineto{\pgfqpoint{1.414136in}{3.575133in}}%
\pgfusepath{stroke}%
\end{pgfscope}%
\begin{pgfscope}%
\pgfsetbuttcap%
\pgfsetroundjoin%
\definecolor{currentfill}{rgb}{0.000000,0.000000,0.000000}%
\pgfsetfillcolor{currentfill}%
\pgfsetlinewidth{0.803000pt}%
\definecolor{currentstroke}{rgb}{0.000000,0.000000,0.000000}%
\pgfsetstrokecolor{currentstroke}%
\pgfsetdash{}{0pt}%
\pgfsys@defobject{currentmarker}{\pgfqpoint{0.000000in}{-0.048611in}}{\pgfqpoint{0.000000in}{0.000000in}}{%
\pgfpathmoveto{\pgfqpoint{0.000000in}{0.000000in}}%
\pgfpathlineto{\pgfqpoint{0.000000in}{-0.048611in}}%
\pgfusepath{stroke,fill}%
}%
\begin{pgfscope}%
\pgfsys@transformshift{1.414136in}{0.494721in}%
\pgfsys@useobject{currentmarker}{}%
\end{pgfscope}%
\end{pgfscope}%
\begin{pgfscope}%
\definecolor{textcolor}{rgb}{0.000000,0.000000,0.000000}%
\pgfsetstrokecolor{textcolor}%
\pgfsetfillcolor{textcolor}%
\pgftext[x=1.414136in,y=0.397499in,,top]{\color{textcolor}\sffamily\fontsize{9.000000}{10.800000}\selectfont \(\displaystyle {20}\)}%
\end{pgfscope}%
\begin{pgfscope}%
\pgfpathrectangle{\pgfqpoint{0.630692in}{0.494721in}}{\pgfqpoint{5.270020in}{3.080412in}}%
\pgfusepath{clip}%
\pgfsetrectcap%
\pgfsetroundjoin%
\pgfsetlinewidth{0.803000pt}%
\definecolor{currentstroke}{rgb}{0.690196,0.690196,0.690196}%
\pgfsetstrokecolor{currentstroke}%
\pgfsetdash{}{0pt}%
\pgfpathmoveto{\pgfqpoint{2.571365in}{0.494721in}}%
\pgfpathlineto{\pgfqpoint{2.571365in}{3.575133in}}%
\pgfusepath{stroke}%
\end{pgfscope}%
\begin{pgfscope}%
\pgfsetbuttcap%
\pgfsetroundjoin%
\definecolor{currentfill}{rgb}{0.000000,0.000000,0.000000}%
\pgfsetfillcolor{currentfill}%
\pgfsetlinewidth{0.803000pt}%
\definecolor{currentstroke}{rgb}{0.000000,0.000000,0.000000}%
\pgfsetstrokecolor{currentstroke}%
\pgfsetdash{}{0pt}%
\pgfsys@defobject{currentmarker}{\pgfqpoint{0.000000in}{-0.048611in}}{\pgfqpoint{0.000000in}{0.000000in}}{%
\pgfpathmoveto{\pgfqpoint{0.000000in}{0.000000in}}%
\pgfpathlineto{\pgfqpoint{0.000000in}{-0.048611in}}%
\pgfusepath{stroke,fill}%
}%
\begin{pgfscope}%
\pgfsys@transformshift{2.571365in}{0.494721in}%
\pgfsys@useobject{currentmarker}{}%
\end{pgfscope}%
\end{pgfscope}%
\begin{pgfscope}%
\definecolor{textcolor}{rgb}{0.000000,0.000000,0.000000}%
\pgfsetstrokecolor{textcolor}%
\pgfsetfillcolor{textcolor}%
\pgftext[x=2.571365in,y=0.397499in,,top]{\color{textcolor}\sffamily\fontsize{9.000000}{10.800000}\selectfont \(\displaystyle {40}\)}%
\end{pgfscope}%
\begin{pgfscope}%
\pgfpathrectangle{\pgfqpoint{0.630692in}{0.494721in}}{\pgfqpoint{5.270020in}{3.080412in}}%
\pgfusepath{clip}%
\pgfsetrectcap%
\pgfsetroundjoin%
\pgfsetlinewidth{0.803000pt}%
\definecolor{currentstroke}{rgb}{0.690196,0.690196,0.690196}%
\pgfsetstrokecolor{currentstroke}%
\pgfsetdash{}{0pt}%
\pgfpathmoveto{\pgfqpoint{3.728594in}{0.494721in}}%
\pgfpathlineto{\pgfqpoint{3.728594in}{3.575133in}}%
\pgfusepath{stroke}%
\end{pgfscope}%
\begin{pgfscope}%
\pgfsetbuttcap%
\pgfsetroundjoin%
\definecolor{currentfill}{rgb}{0.000000,0.000000,0.000000}%
\pgfsetfillcolor{currentfill}%
\pgfsetlinewidth{0.803000pt}%
\definecolor{currentstroke}{rgb}{0.000000,0.000000,0.000000}%
\pgfsetstrokecolor{currentstroke}%
\pgfsetdash{}{0pt}%
\pgfsys@defobject{currentmarker}{\pgfqpoint{0.000000in}{-0.048611in}}{\pgfqpoint{0.000000in}{0.000000in}}{%
\pgfpathmoveto{\pgfqpoint{0.000000in}{0.000000in}}%
\pgfpathlineto{\pgfqpoint{0.000000in}{-0.048611in}}%
\pgfusepath{stroke,fill}%
}%
\begin{pgfscope}%
\pgfsys@transformshift{3.728594in}{0.494721in}%
\pgfsys@useobject{currentmarker}{}%
\end{pgfscope}%
\end{pgfscope}%
\begin{pgfscope}%
\definecolor{textcolor}{rgb}{0.000000,0.000000,0.000000}%
\pgfsetstrokecolor{textcolor}%
\pgfsetfillcolor{textcolor}%
\pgftext[x=3.728594in,y=0.397499in,,top]{\color{textcolor}\sffamily\fontsize{9.000000}{10.800000}\selectfont \(\displaystyle {60}\)}%
\end{pgfscope}%
\begin{pgfscope}%
\pgfpathrectangle{\pgfqpoint{0.630692in}{0.494721in}}{\pgfqpoint{5.270020in}{3.080412in}}%
\pgfusepath{clip}%
\pgfsetrectcap%
\pgfsetroundjoin%
\pgfsetlinewidth{0.803000pt}%
\definecolor{currentstroke}{rgb}{0.690196,0.690196,0.690196}%
\pgfsetstrokecolor{currentstroke}%
\pgfsetdash{}{0pt}%
\pgfpathmoveto{\pgfqpoint{4.885822in}{0.494721in}}%
\pgfpathlineto{\pgfqpoint{4.885822in}{3.575133in}}%
\pgfusepath{stroke}%
\end{pgfscope}%
\begin{pgfscope}%
\pgfsetbuttcap%
\pgfsetroundjoin%
\definecolor{currentfill}{rgb}{0.000000,0.000000,0.000000}%
\pgfsetfillcolor{currentfill}%
\pgfsetlinewidth{0.803000pt}%
\definecolor{currentstroke}{rgb}{0.000000,0.000000,0.000000}%
\pgfsetstrokecolor{currentstroke}%
\pgfsetdash{}{0pt}%
\pgfsys@defobject{currentmarker}{\pgfqpoint{0.000000in}{-0.048611in}}{\pgfqpoint{0.000000in}{0.000000in}}{%
\pgfpathmoveto{\pgfqpoint{0.000000in}{0.000000in}}%
\pgfpathlineto{\pgfqpoint{0.000000in}{-0.048611in}}%
\pgfusepath{stroke,fill}%
}%
\begin{pgfscope}%
\pgfsys@transformshift{4.885822in}{0.494721in}%
\pgfsys@useobject{currentmarker}{}%
\end{pgfscope}%
\end{pgfscope}%
\begin{pgfscope}%
\definecolor{textcolor}{rgb}{0.000000,0.000000,0.000000}%
\pgfsetstrokecolor{textcolor}%
\pgfsetfillcolor{textcolor}%
\pgftext[x=4.885822in,y=0.397499in,,top]{\color{textcolor}\sffamily\fontsize{9.000000}{10.800000}\selectfont \(\displaystyle {80}\)}%
\end{pgfscope}%
\begin{pgfscope}%
\definecolor{textcolor}{rgb}{0.000000,0.000000,0.000000}%
\pgfsetstrokecolor{textcolor}%
\pgfsetfillcolor{textcolor}%
\pgftext[x=3.265702in,y=0.220972in,,top]{\color{textcolor}\sffamily\fontsize{9.000000}{10.800000}\selectfont Mean of 10 project-join tree widths}%
\end{pgfscope}%
\begin{pgfscope}%
\pgfpathrectangle{\pgfqpoint{0.630692in}{0.494721in}}{\pgfqpoint{5.270020in}{3.080412in}}%
\pgfusepath{clip}%
\pgfsetrectcap%
\pgfsetroundjoin%
\pgfsetlinewidth{0.803000pt}%
\definecolor{currentstroke}{rgb}{0.690196,0.690196,0.690196}%
\pgfsetstrokecolor{currentstroke}%
\pgfsetdash{}{0pt}%
\pgfpathmoveto{\pgfqpoint{0.630692in}{0.634697in}}%
\pgfpathlineto{\pgfqpoint{5.900712in}{0.634697in}}%
\pgfusepath{stroke}%
\end{pgfscope}%
\begin{pgfscope}%
\pgfsetbuttcap%
\pgfsetroundjoin%
\definecolor{currentfill}{rgb}{0.000000,0.000000,0.000000}%
\pgfsetfillcolor{currentfill}%
\pgfsetlinewidth{0.803000pt}%
\definecolor{currentstroke}{rgb}{0.000000,0.000000,0.000000}%
\pgfsetstrokecolor{currentstroke}%
\pgfsetdash{}{0pt}%
\pgfsys@defobject{currentmarker}{\pgfqpoint{-0.048611in}{0.000000in}}{\pgfqpoint{-0.000000in}{0.000000in}}{%
\pgfpathmoveto{\pgfqpoint{-0.000000in}{0.000000in}}%
\pgfpathlineto{\pgfqpoint{-0.048611in}{0.000000in}}%
\pgfusepath{stroke,fill}%
}%
\begin{pgfscope}%
\pgfsys@transformshift{0.630692in}{0.634697in}%
\pgfsys@useobject{currentmarker}{}%
\end{pgfscope}%
\end{pgfscope}%
\begin{pgfscope}%
\definecolor{textcolor}{rgb}{0.000000,0.000000,0.000000}%
\pgfsetstrokecolor{textcolor}%
\pgfsetfillcolor{textcolor}%
\pgftext[x=0.469234in, y=0.587212in, left, base]{\color{textcolor}\sffamily\fontsize{9.000000}{10.800000}\selectfont \(\displaystyle {0}\)}%
\end{pgfscope}%
\begin{pgfscope}%
\pgfpathrectangle{\pgfqpoint{0.630692in}{0.494721in}}{\pgfqpoint{5.270020in}{3.080412in}}%
\pgfusepath{clip}%
\pgfsetrectcap%
\pgfsetroundjoin%
\pgfsetlinewidth{0.803000pt}%
\definecolor{currentstroke}{rgb}{0.690196,0.690196,0.690196}%
\pgfsetstrokecolor{currentstroke}%
\pgfsetdash{}{0pt}%
\pgfpathmoveto{\pgfqpoint{0.630692in}{1.047822in}}%
\pgfpathlineto{\pgfqpoint{5.900712in}{1.047822in}}%
\pgfusepath{stroke}%
\end{pgfscope}%
\begin{pgfscope}%
\pgfsetbuttcap%
\pgfsetroundjoin%
\definecolor{currentfill}{rgb}{0.000000,0.000000,0.000000}%
\pgfsetfillcolor{currentfill}%
\pgfsetlinewidth{0.803000pt}%
\definecolor{currentstroke}{rgb}{0.000000,0.000000,0.000000}%
\pgfsetstrokecolor{currentstroke}%
\pgfsetdash{}{0pt}%
\pgfsys@defobject{currentmarker}{\pgfqpoint{-0.048611in}{0.000000in}}{\pgfqpoint{-0.000000in}{0.000000in}}{%
\pgfpathmoveto{\pgfqpoint{-0.000000in}{0.000000in}}%
\pgfpathlineto{\pgfqpoint{-0.048611in}{0.000000in}}%
\pgfusepath{stroke,fill}%
}%
\begin{pgfscope}%
\pgfsys@transformshift{0.630692in}{1.047822in}%
\pgfsys@useobject{currentmarker}{}%
\end{pgfscope}%
\end{pgfscope}%
\begin{pgfscope}%
\definecolor{textcolor}{rgb}{0.000000,0.000000,0.000000}%
\pgfsetstrokecolor{textcolor}%
\pgfsetfillcolor{textcolor}%
\pgftext[x=0.340763in, y=1.000336in, left, base]{\color{textcolor}\sffamily\fontsize{9.000000}{10.800000}\selectfont \(\displaystyle {250}\)}%
\end{pgfscope}%
\begin{pgfscope}%
\pgfpathrectangle{\pgfqpoint{0.630692in}{0.494721in}}{\pgfqpoint{5.270020in}{3.080412in}}%
\pgfusepath{clip}%
\pgfsetrectcap%
\pgfsetroundjoin%
\pgfsetlinewidth{0.803000pt}%
\definecolor{currentstroke}{rgb}{0.690196,0.690196,0.690196}%
\pgfsetstrokecolor{currentstroke}%
\pgfsetdash{}{0pt}%
\pgfpathmoveto{\pgfqpoint{0.630692in}{1.460946in}}%
\pgfpathlineto{\pgfqpoint{5.900712in}{1.460946in}}%
\pgfusepath{stroke}%
\end{pgfscope}%
\begin{pgfscope}%
\pgfsetbuttcap%
\pgfsetroundjoin%
\definecolor{currentfill}{rgb}{0.000000,0.000000,0.000000}%
\pgfsetfillcolor{currentfill}%
\pgfsetlinewidth{0.803000pt}%
\definecolor{currentstroke}{rgb}{0.000000,0.000000,0.000000}%
\pgfsetstrokecolor{currentstroke}%
\pgfsetdash{}{0pt}%
\pgfsys@defobject{currentmarker}{\pgfqpoint{-0.048611in}{0.000000in}}{\pgfqpoint{-0.000000in}{0.000000in}}{%
\pgfpathmoveto{\pgfqpoint{-0.000000in}{0.000000in}}%
\pgfpathlineto{\pgfqpoint{-0.048611in}{0.000000in}}%
\pgfusepath{stroke,fill}%
}%
\begin{pgfscope}%
\pgfsys@transformshift{0.630692in}{1.460946in}%
\pgfsys@useobject{currentmarker}{}%
\end{pgfscope}%
\end{pgfscope}%
\begin{pgfscope}%
\definecolor{textcolor}{rgb}{0.000000,0.000000,0.000000}%
\pgfsetstrokecolor{textcolor}%
\pgfsetfillcolor{textcolor}%
\pgftext[x=0.340763in, y=1.413460in, left, base]{\color{textcolor}\sffamily\fontsize{9.000000}{10.800000}\selectfont \(\displaystyle {500}\)}%
\end{pgfscope}%
\begin{pgfscope}%
\pgfpathrectangle{\pgfqpoint{0.630692in}{0.494721in}}{\pgfqpoint{5.270020in}{3.080412in}}%
\pgfusepath{clip}%
\pgfsetrectcap%
\pgfsetroundjoin%
\pgfsetlinewidth{0.803000pt}%
\definecolor{currentstroke}{rgb}{0.690196,0.690196,0.690196}%
\pgfsetstrokecolor{currentstroke}%
\pgfsetdash{}{0pt}%
\pgfpathmoveto{\pgfqpoint{0.630692in}{1.874070in}}%
\pgfpathlineto{\pgfqpoint{5.900712in}{1.874070in}}%
\pgfusepath{stroke}%
\end{pgfscope}%
\begin{pgfscope}%
\pgfsetbuttcap%
\pgfsetroundjoin%
\definecolor{currentfill}{rgb}{0.000000,0.000000,0.000000}%
\pgfsetfillcolor{currentfill}%
\pgfsetlinewidth{0.803000pt}%
\definecolor{currentstroke}{rgb}{0.000000,0.000000,0.000000}%
\pgfsetstrokecolor{currentstroke}%
\pgfsetdash{}{0pt}%
\pgfsys@defobject{currentmarker}{\pgfqpoint{-0.048611in}{0.000000in}}{\pgfqpoint{-0.000000in}{0.000000in}}{%
\pgfpathmoveto{\pgfqpoint{-0.000000in}{0.000000in}}%
\pgfpathlineto{\pgfqpoint{-0.048611in}{0.000000in}}%
\pgfusepath{stroke,fill}%
}%
\begin{pgfscope}%
\pgfsys@transformshift{0.630692in}{1.874070in}%
\pgfsys@useobject{currentmarker}{}%
\end{pgfscope}%
\end{pgfscope}%
\begin{pgfscope}%
\definecolor{textcolor}{rgb}{0.000000,0.000000,0.000000}%
\pgfsetstrokecolor{textcolor}%
\pgfsetfillcolor{textcolor}%
\pgftext[x=0.340763in, y=1.826584in, left, base]{\color{textcolor}\sffamily\fontsize{9.000000}{10.800000}\selectfont \(\displaystyle {750}\)}%
\end{pgfscope}%
\begin{pgfscope}%
\pgfpathrectangle{\pgfqpoint{0.630692in}{0.494721in}}{\pgfqpoint{5.270020in}{3.080412in}}%
\pgfusepath{clip}%
\pgfsetrectcap%
\pgfsetroundjoin%
\pgfsetlinewidth{0.803000pt}%
\definecolor{currentstroke}{rgb}{0.690196,0.690196,0.690196}%
\pgfsetstrokecolor{currentstroke}%
\pgfsetdash{}{0pt}%
\pgfpathmoveto{\pgfqpoint{0.630692in}{2.287194in}}%
\pgfpathlineto{\pgfqpoint{5.900712in}{2.287194in}}%
\pgfusepath{stroke}%
\end{pgfscope}%
\begin{pgfscope}%
\pgfsetbuttcap%
\pgfsetroundjoin%
\definecolor{currentfill}{rgb}{0.000000,0.000000,0.000000}%
\pgfsetfillcolor{currentfill}%
\pgfsetlinewidth{0.803000pt}%
\definecolor{currentstroke}{rgb}{0.000000,0.000000,0.000000}%
\pgfsetstrokecolor{currentstroke}%
\pgfsetdash{}{0pt}%
\pgfsys@defobject{currentmarker}{\pgfqpoint{-0.048611in}{0.000000in}}{\pgfqpoint{-0.000000in}{0.000000in}}{%
\pgfpathmoveto{\pgfqpoint{-0.000000in}{0.000000in}}%
\pgfpathlineto{\pgfqpoint{-0.048611in}{0.000000in}}%
\pgfusepath{stroke,fill}%
}%
\begin{pgfscope}%
\pgfsys@transformshift{0.630692in}{2.287194in}%
\pgfsys@useobject{currentmarker}{}%
\end{pgfscope}%
\end{pgfscope}%
\begin{pgfscope}%
\definecolor{textcolor}{rgb}{0.000000,0.000000,0.000000}%
\pgfsetstrokecolor{textcolor}%
\pgfsetfillcolor{textcolor}%
\pgftext[x=0.276527in, y=2.239709in, left, base]{\color{textcolor}\sffamily\fontsize{9.000000}{10.800000}\selectfont \(\displaystyle {1000}\)}%
\end{pgfscope}%
\begin{pgfscope}%
\pgfpathrectangle{\pgfqpoint{0.630692in}{0.494721in}}{\pgfqpoint{5.270020in}{3.080412in}}%
\pgfusepath{clip}%
\pgfsetrectcap%
\pgfsetroundjoin%
\pgfsetlinewidth{0.803000pt}%
\definecolor{currentstroke}{rgb}{0.690196,0.690196,0.690196}%
\pgfsetstrokecolor{currentstroke}%
\pgfsetdash{}{0pt}%
\pgfpathmoveto{\pgfqpoint{0.630692in}{2.700318in}}%
\pgfpathlineto{\pgfqpoint{5.900712in}{2.700318in}}%
\pgfusepath{stroke}%
\end{pgfscope}%
\begin{pgfscope}%
\pgfsetbuttcap%
\pgfsetroundjoin%
\definecolor{currentfill}{rgb}{0.000000,0.000000,0.000000}%
\pgfsetfillcolor{currentfill}%
\pgfsetlinewidth{0.803000pt}%
\definecolor{currentstroke}{rgb}{0.000000,0.000000,0.000000}%
\pgfsetstrokecolor{currentstroke}%
\pgfsetdash{}{0pt}%
\pgfsys@defobject{currentmarker}{\pgfqpoint{-0.048611in}{0.000000in}}{\pgfqpoint{-0.000000in}{0.000000in}}{%
\pgfpathmoveto{\pgfqpoint{-0.000000in}{0.000000in}}%
\pgfpathlineto{\pgfqpoint{-0.048611in}{0.000000in}}%
\pgfusepath{stroke,fill}%
}%
\begin{pgfscope}%
\pgfsys@transformshift{0.630692in}{2.700318in}%
\pgfsys@useobject{currentmarker}{}%
\end{pgfscope}%
\end{pgfscope}%
\begin{pgfscope}%
\definecolor{textcolor}{rgb}{0.000000,0.000000,0.000000}%
\pgfsetstrokecolor{textcolor}%
\pgfsetfillcolor{textcolor}%
\pgftext[x=0.276527in, y=2.652833in, left, base]{\color{textcolor}\sffamily\fontsize{9.000000}{10.800000}\selectfont \(\displaystyle {1250}\)}%
\end{pgfscope}%
\begin{pgfscope}%
\pgfpathrectangle{\pgfqpoint{0.630692in}{0.494721in}}{\pgfqpoint{5.270020in}{3.080412in}}%
\pgfusepath{clip}%
\pgfsetrectcap%
\pgfsetroundjoin%
\pgfsetlinewidth{0.803000pt}%
\definecolor{currentstroke}{rgb}{0.690196,0.690196,0.690196}%
\pgfsetstrokecolor{currentstroke}%
\pgfsetdash{}{0pt}%
\pgfpathmoveto{\pgfqpoint{0.630692in}{3.113442in}}%
\pgfpathlineto{\pgfqpoint{5.900712in}{3.113442in}}%
\pgfusepath{stroke}%
\end{pgfscope}%
\begin{pgfscope}%
\pgfsetbuttcap%
\pgfsetroundjoin%
\definecolor{currentfill}{rgb}{0.000000,0.000000,0.000000}%
\pgfsetfillcolor{currentfill}%
\pgfsetlinewidth{0.803000pt}%
\definecolor{currentstroke}{rgb}{0.000000,0.000000,0.000000}%
\pgfsetstrokecolor{currentstroke}%
\pgfsetdash{}{0pt}%
\pgfsys@defobject{currentmarker}{\pgfqpoint{-0.048611in}{0.000000in}}{\pgfqpoint{-0.000000in}{0.000000in}}{%
\pgfpathmoveto{\pgfqpoint{-0.000000in}{0.000000in}}%
\pgfpathlineto{\pgfqpoint{-0.048611in}{0.000000in}}%
\pgfusepath{stroke,fill}%
}%
\begin{pgfscope}%
\pgfsys@transformshift{0.630692in}{3.113442in}%
\pgfsys@useobject{currentmarker}{}%
\end{pgfscope}%
\end{pgfscope}%
\begin{pgfscope}%
\definecolor{textcolor}{rgb}{0.000000,0.000000,0.000000}%
\pgfsetstrokecolor{textcolor}%
\pgfsetfillcolor{textcolor}%
\pgftext[x=0.276527in, y=3.065957in, left, base]{\color{textcolor}\sffamily\fontsize{9.000000}{10.800000}\selectfont \(\displaystyle {1500}\)}%
\end{pgfscope}%
\begin{pgfscope}%
\pgfpathrectangle{\pgfqpoint{0.630692in}{0.494721in}}{\pgfqpoint{5.270020in}{3.080412in}}%
\pgfusepath{clip}%
\pgfsetrectcap%
\pgfsetroundjoin%
\pgfsetlinewidth{0.803000pt}%
\definecolor{currentstroke}{rgb}{0.690196,0.690196,0.690196}%
\pgfsetstrokecolor{currentstroke}%
\pgfsetdash{}{0pt}%
\pgfpathmoveto{\pgfqpoint{0.630692in}{3.526566in}}%
\pgfpathlineto{\pgfqpoint{5.900712in}{3.526566in}}%
\pgfusepath{stroke}%
\end{pgfscope}%
\begin{pgfscope}%
\pgfsetbuttcap%
\pgfsetroundjoin%
\definecolor{currentfill}{rgb}{0.000000,0.000000,0.000000}%
\pgfsetfillcolor{currentfill}%
\pgfsetlinewidth{0.803000pt}%
\definecolor{currentstroke}{rgb}{0.000000,0.000000,0.000000}%
\pgfsetstrokecolor{currentstroke}%
\pgfsetdash{}{0pt}%
\pgfsys@defobject{currentmarker}{\pgfqpoint{-0.048611in}{0.000000in}}{\pgfqpoint{-0.000000in}{0.000000in}}{%
\pgfpathmoveto{\pgfqpoint{-0.000000in}{0.000000in}}%
\pgfpathlineto{\pgfqpoint{-0.048611in}{0.000000in}}%
\pgfusepath{stroke,fill}%
}%
\begin{pgfscope}%
\pgfsys@transformshift{0.630692in}{3.526566in}%
\pgfsys@useobject{currentmarker}{}%
\end{pgfscope}%
\end{pgfscope}%
\begin{pgfscope}%
\definecolor{textcolor}{rgb}{0.000000,0.000000,0.000000}%
\pgfsetstrokecolor{textcolor}%
\pgfsetfillcolor{textcolor}%
\pgftext[x=0.276527in, y=3.479081in, left, base]{\color{textcolor}\sffamily\fontsize{9.000000}{10.800000}\selectfont \(\displaystyle {1750}\)}%
\end{pgfscope}%
\begin{pgfscope}%
\definecolor{textcolor}{rgb}{0.000000,0.000000,0.000000}%
\pgfsetstrokecolor{textcolor}%
\pgfsetfillcolor{textcolor}%
\pgftext[x=0.220972in,y=2.034927in,,bottom,rotate=90.000000]{\color{textcolor}\sffamily\fontsize{9.000000}{10.800000}\selectfont Mean PAR-2 score of 10 widths}%
\end{pgfscope}%
\begin{pgfscope}%
\pgfpathrectangle{\pgfqpoint{0.630692in}{0.494721in}}{\pgfqpoint{5.270020in}{3.080412in}}%
\pgfusepath{clip}%
\pgfsetbuttcap%
\pgfsetroundjoin%
\pgfsetlinewidth{1.003750pt}%
\definecolor{currentstroke}{rgb}{0.121569,0.466667,0.705882}%
\pgfsetstrokecolor{currentstroke}%
\pgfsetdash{{3.700000pt}{1.600000pt}}{0.000000pt}%
\pgfpathmoveto{\pgfqpoint{0.870238in}{0.634740in}}%
\pgfpathlineto{\pgfqpoint{0.974389in}{0.652357in}}%
\pgfpathlineto{\pgfqpoint{1.066967in}{0.681903in}}%
\pgfpathlineto{\pgfqpoint{1.159546in}{0.705071in}}%
\pgfpathlineto{\pgfqpoint{1.252124in}{0.701368in}}%
\pgfpathlineto{\pgfqpoint{1.327344in}{0.863395in}}%
\pgfpathlineto{\pgfqpoint{1.402564in}{1.019036in}}%
\pgfpathlineto{\pgfqpoint{1.466211in}{0.963182in}}%
\pgfpathlineto{\pgfqpoint{1.547217in}{1.194119in}}%
\pgfpathlineto{\pgfqpoint{1.628223in}{1.194126in}}%
\pgfpathlineto{\pgfqpoint{1.709229in}{1.167507in}}%
\pgfpathlineto{\pgfqpoint{1.807594in}{1.021505in}}%
\pgfpathlineto{\pgfqpoint{1.905958in}{1.008299in}}%
\pgfpathlineto{\pgfqpoint{2.010109in}{1.016674in}}%
\pgfpathlineto{\pgfqpoint{2.114259in}{1.040058in}}%
\pgfpathlineto{\pgfqpoint{2.218410in}{0.910746in}}%
\pgfpathlineto{\pgfqpoint{2.328347in}{0.760056in}}%
\pgfpathlineto{\pgfqpoint{2.438284in}{0.775861in}}%
\pgfpathlineto{\pgfqpoint{2.536648in}{0.635261in}}%
\pgfpathlineto{\pgfqpoint{2.635012in}{0.679385in}}%
\pgfpathlineto{\pgfqpoint{2.733377in}{0.823835in}}%
\pgfpathlineto{\pgfqpoint{2.814383in}{0.934515in}}%
\pgfpathlineto{\pgfqpoint{2.895389in}{0.925257in}}%
\pgfpathlineto{\pgfqpoint{2.970609in}{1.061486in}}%
\pgfpathlineto{\pgfqpoint{3.040042in}{1.139282in}}%
\pgfpathlineto{\pgfqpoint{3.115262in}{1.336850in}}%
\pgfpathlineto{\pgfqpoint{3.184696in}{1.425421in}}%
\pgfpathlineto{\pgfqpoint{3.254130in}{1.425720in}}%
\pgfpathlineto{\pgfqpoint{3.317777in}{1.735618in}}%
\pgfpathlineto{\pgfqpoint{3.381425in}{1.484298in}}%
\pgfpathlineto{\pgfqpoint{3.445073in}{1.363657in}}%
\pgfpathlineto{\pgfqpoint{3.508720in}{1.424991in}}%
\pgfpathlineto{\pgfqpoint{3.583940in}{1.355804in}}%
\pgfpathlineto{\pgfqpoint{3.676518in}{1.462749in}}%
\pgfpathlineto{\pgfqpoint{3.769097in}{1.399234in}}%
\pgfpathlineto{\pgfqpoint{3.855889in}{1.229750in}}%
\pgfpathlineto{\pgfqpoint{3.942681in}{1.483229in}}%
\pgfpathlineto{\pgfqpoint{4.029473in}{1.711513in}}%
\pgfpathlineto{\pgfqpoint{4.116265in}{1.492403in}}%
\pgfpathlineto{\pgfqpoint{4.214630in}{1.891575in}}%
\pgfpathlineto{\pgfqpoint{4.312994in}{2.068899in}}%
\pgfpathlineto{\pgfqpoint{4.411359in}{2.113716in}}%
\pgfpathlineto{\pgfqpoint{4.503937in}{2.763126in}}%
\pgfpathlineto{\pgfqpoint{4.579157in}{2.694471in}}%
\pgfpathlineto{\pgfqpoint{4.654377in}{2.801081in}}%
\pgfpathlineto{\pgfqpoint{4.729597in}{2.873292in}}%
\pgfpathlineto{\pgfqpoint{4.804816in}{2.679100in}}%
\pgfpathlineto{\pgfqpoint{4.891609in}{2.548753in}}%
\pgfpathlineto{\pgfqpoint{4.978401in}{2.646064in}}%
\pgfpathlineto{\pgfqpoint{5.053621in}{2.756037in}}%
\pgfpathlineto{\pgfqpoint{5.128841in}{2.781855in}}%
\pgfpathlineto{\pgfqpoint{5.204060in}{2.903732in}}%
\pgfpathlineto{\pgfqpoint{5.273494in}{2.408547in}}%
\pgfpathlineto{\pgfqpoint{5.348714in}{2.783086in}}%
\pgfpathlineto{\pgfqpoint{5.423934in}{2.295693in}}%
\pgfpathlineto{\pgfqpoint{5.504940in}{2.653820in}}%
\pgfpathlineto{\pgfqpoint{5.585946in}{2.888873in}}%
\pgfpathlineto{\pgfqpoint{5.661166in}{2.844356in}}%
\pgfusepath{stroke}%
\end{pgfscope}%
\begin{pgfscope}%
\pgfpathrectangle{\pgfqpoint{0.630692in}{0.494721in}}{\pgfqpoint{5.270020in}{3.080412in}}%
\pgfusepath{clip}%
\pgfsetbuttcap%
\pgfsetroundjoin%
\pgfsetlinewidth{1.003750pt}%
\definecolor{currentstroke}{rgb}{1.000000,0.498039,0.054902}%
\pgfsetstrokecolor{currentstroke}%
\pgfsetdash{{3.700000pt}{1.600000pt}}{0.000000pt}%
\pgfpathmoveto{\pgfqpoint{0.870238in}{1.821179in}}%
\pgfpathlineto{\pgfqpoint{0.974389in}{2.085737in}}%
\pgfpathlineto{\pgfqpoint{1.066967in}{2.767010in}}%
\pgfpathlineto{\pgfqpoint{1.159546in}{3.021686in}}%
\pgfpathlineto{\pgfqpoint{1.252124in}{3.096765in}}%
\pgfpathlineto{\pgfqpoint{1.327344in}{3.138917in}}%
\pgfpathlineto{\pgfqpoint{1.402564in}{3.077847in}}%
\pgfpathlineto{\pgfqpoint{1.466211in}{3.237670in}}%
\pgfpathlineto{\pgfqpoint{1.547217in}{3.179879in}}%
\pgfpathlineto{\pgfqpoint{1.628223in}{3.179879in}}%
\pgfpathlineto{\pgfqpoint{1.709229in}{3.216055in}}%
\pgfpathlineto{\pgfqpoint{1.807594in}{3.396964in}}%
\pgfpathlineto{\pgfqpoint{1.905958in}{3.396964in}}%
\pgfpathlineto{\pgfqpoint{2.010109in}{3.301407in}}%
\pgfpathlineto{\pgfqpoint{2.114259in}{3.309022in}}%
\pgfpathlineto{\pgfqpoint{2.218410in}{3.297815in}}%
\pgfpathlineto{\pgfqpoint{2.328347in}{3.435114in}}%
\pgfpathlineto{\pgfqpoint{2.438284in}{3.336435in}}%
\pgfpathlineto{\pgfqpoint{2.536648in}{3.372046in}}%
\pgfpathlineto{\pgfqpoint{2.635012in}{2.965674in}}%
\pgfpathlineto{\pgfqpoint{2.733377in}{2.965692in}}%
\pgfpathlineto{\pgfqpoint{2.814383in}{2.280850in}}%
\pgfpathlineto{\pgfqpoint{2.895389in}{2.336144in}}%
\pgfpathlineto{\pgfqpoint{2.970609in}{2.141304in}}%
\pgfpathlineto{\pgfqpoint{3.040042in}{2.118530in}}%
\pgfpathlineto{\pgfqpoint{3.115262in}{2.075441in}}%
\pgfpathlineto{\pgfqpoint{3.184696in}{1.921846in}}%
\pgfpathlineto{\pgfqpoint{3.254130in}{2.024244in}}%
\pgfpathlineto{\pgfqpoint{3.317777in}{2.103672in}}%
\pgfpathlineto{\pgfqpoint{3.381425in}{2.906382in}}%
\pgfpathlineto{\pgfqpoint{3.445073in}{2.936564in}}%
\pgfpathlineto{\pgfqpoint{3.508720in}{2.960418in}}%
\pgfpathlineto{\pgfqpoint{3.583940in}{2.973735in}}%
\pgfpathlineto{\pgfqpoint{3.676518in}{3.150045in}}%
\pgfpathlineto{\pgfqpoint{3.769097in}{3.197517in}}%
\pgfpathlineto{\pgfqpoint{3.855889in}{3.357696in}}%
\pgfpathlineto{\pgfqpoint{3.942681in}{3.272988in}}%
\pgfpathlineto{\pgfqpoint{4.029473in}{3.028224in}}%
\pgfpathlineto{\pgfqpoint{4.116265in}{3.096580in}}%
\pgfpathlineto{\pgfqpoint{4.214630in}{2.703544in}}%
\pgfpathlineto{\pgfqpoint{4.312994in}{2.642025in}}%
\pgfpathlineto{\pgfqpoint{4.411359in}{2.602782in}}%
\pgfpathlineto{\pgfqpoint{4.503937in}{2.403570in}}%
\pgfpathlineto{\pgfqpoint{4.579157in}{2.346049in}}%
\pgfpathlineto{\pgfqpoint{4.654377in}{2.452658in}}%
\pgfpathlineto{\pgfqpoint{4.729597in}{2.293245in}}%
\pgfpathlineto{\pgfqpoint{4.804816in}{2.176179in}}%
\pgfpathlineto{\pgfqpoint{4.891609in}{2.366356in}}%
\pgfpathlineto{\pgfqpoint{4.978401in}{2.131070in}}%
\pgfpathlineto{\pgfqpoint{5.053621in}{1.861960in}}%
\pgfpathlineto{\pgfqpoint{5.128841in}{1.619996in}}%
\pgfpathlineto{\pgfqpoint{5.204060in}{1.516774in}}%
\pgfpathlineto{\pgfqpoint{5.273494in}{1.311134in}}%
\pgfpathlineto{\pgfqpoint{5.348714in}{1.426503in}}%
\pgfpathlineto{\pgfqpoint{5.423934in}{1.803481in}}%
\pgfpathlineto{\pgfqpoint{5.504940in}{2.119988in}}%
\pgfpathlineto{\pgfqpoint{5.585946in}{2.133066in}}%
\pgfpathlineto{\pgfqpoint{5.661166in}{2.144533in}}%
\pgfusepath{stroke}%
\end{pgfscope}%
\begin{pgfscope}%
\pgfpathrectangle{\pgfqpoint{0.630692in}{0.494721in}}{\pgfqpoint{5.270020in}{3.080412in}}%
\pgfusepath{clip}%
\pgfsetrectcap%
\pgfsetroundjoin%
\pgfsetlinewidth{1.003750pt}%
\definecolor{currentstroke}{rgb}{0.172549,0.627451,0.172549}%
\pgfsetstrokecolor{currentstroke}%
\pgfsetdash{}{0pt}%
\pgfpathmoveto{\pgfqpoint{0.870238in}{0.635119in}}%
\pgfpathlineto{\pgfqpoint{0.974389in}{0.635121in}}%
\pgfpathlineto{\pgfqpoint{1.066967in}{0.635124in}}%
\pgfpathlineto{\pgfqpoint{1.159546in}{0.635141in}}%
\pgfpathlineto{\pgfqpoint{1.252124in}{0.635138in}}%
\pgfpathlineto{\pgfqpoint{1.327344in}{0.635154in}}%
\pgfpathlineto{\pgfqpoint{1.402564in}{0.635165in}}%
\pgfpathlineto{\pgfqpoint{1.466211in}{0.635163in}}%
\pgfpathlineto{\pgfqpoint{1.547217in}{0.635248in}}%
\pgfpathlineto{\pgfqpoint{1.628223in}{0.635251in}}%
\pgfpathlineto{\pgfqpoint{1.709229in}{0.635247in}}%
\pgfpathlineto{\pgfqpoint{1.807594in}{0.635217in}}%
\pgfpathlineto{\pgfqpoint{1.905958in}{0.635216in}}%
\pgfpathlineto{\pgfqpoint{2.010109in}{0.635214in}}%
\pgfpathlineto{\pgfqpoint{2.114259in}{0.635222in}}%
\pgfpathlineto{\pgfqpoint{2.218410in}{0.635218in}}%
\pgfpathlineto{\pgfqpoint{2.328347in}{0.635204in}}%
\pgfpathlineto{\pgfqpoint{2.438284in}{0.635409in}}%
\pgfpathlineto{\pgfqpoint{2.536648in}{0.635657in}}%
\pgfpathlineto{\pgfqpoint{2.635012in}{0.635640in}}%
\pgfpathlineto{\pgfqpoint{2.733377in}{0.707699in}}%
\pgfpathlineto{\pgfqpoint{2.814383in}{0.750240in}}%
\pgfpathlineto{\pgfqpoint{2.895389in}{0.746446in}}%
\pgfpathlineto{\pgfqpoint{2.970609in}{0.913034in}}%
\pgfpathlineto{\pgfqpoint{3.040042in}{0.905679in}}%
\pgfpathlineto{\pgfqpoint{3.115262in}{1.006028in}}%
\pgfpathlineto{\pgfqpoint{3.184696in}{1.052397in}}%
\pgfpathlineto{\pgfqpoint{3.254130in}{1.052179in}}%
\pgfpathlineto{\pgfqpoint{3.317777in}{1.114377in}}%
\pgfpathlineto{\pgfqpoint{3.381425in}{1.025401in}}%
\pgfpathlineto{\pgfqpoint{3.445073in}{0.963657in}}%
\pgfpathlineto{\pgfqpoint{3.508720in}{0.955860in}}%
\pgfpathlineto{\pgfqpoint{3.583940in}{0.927956in}}%
\pgfpathlineto{\pgfqpoint{3.676518in}{0.970478in}}%
\pgfpathlineto{\pgfqpoint{3.769097in}{0.980231in}}%
\pgfpathlineto{\pgfqpoint{3.855889in}{0.894537in}}%
\pgfpathlineto{\pgfqpoint{3.942681in}{0.894846in}}%
\pgfpathlineto{\pgfqpoint{4.029473in}{1.102312in}}%
\pgfpathlineto{\pgfqpoint{4.116265in}{1.023554in}}%
\pgfpathlineto{\pgfqpoint{4.214630in}{1.466616in}}%
\pgfpathlineto{\pgfqpoint{4.312994in}{1.547114in}}%
\pgfpathlineto{\pgfqpoint{4.411359in}{1.579375in}}%
\pgfpathlineto{\pgfqpoint{4.503937in}{1.862776in}}%
\pgfpathlineto{\pgfqpoint{4.579157in}{1.716802in}}%
\pgfpathlineto{\pgfqpoint{4.654377in}{1.821761in}}%
\pgfpathlineto{\pgfqpoint{4.729597in}{1.861313in}}%
\pgfpathlineto{\pgfqpoint{4.804816in}{1.952557in}}%
\pgfpathlineto{\pgfqpoint{4.891609in}{1.805937in}}%
\pgfpathlineto{\pgfqpoint{4.978401in}{2.144062in}}%
\pgfpathlineto{\pgfqpoint{5.053621in}{2.049556in}}%
\pgfpathlineto{\pgfqpoint{5.128841in}{2.205016in}}%
\pgfpathlineto{\pgfqpoint{5.204060in}{2.383067in}}%
\pgfpathlineto{\pgfqpoint{5.273494in}{2.628539in}}%
\pgfpathlineto{\pgfqpoint{5.348714in}{3.171151in}}%
\pgfpathlineto{\pgfqpoint{5.423934in}{2.638210in}}%
\pgfpathlineto{\pgfqpoint{5.504940in}{2.951635in}}%
\pgfpathlineto{\pgfqpoint{5.585946in}{3.158268in}}%
\pgfpathlineto{\pgfqpoint{5.661166in}{3.107034in}}%
\pgfusepath{stroke}%
\end{pgfscope}%
\begin{pgfscope}%
\pgfsetrectcap%
\pgfsetmiterjoin%
\pgfsetlinewidth{0.803000pt}%
\definecolor{currentstroke}{rgb}{0.000000,0.000000,0.000000}%
\pgfsetstrokecolor{currentstroke}%
\pgfsetdash{}{0pt}%
\pgfpathmoveto{\pgfqpoint{0.630692in}{0.494721in}}%
\pgfpathlineto{\pgfqpoint{0.630692in}{3.575133in}}%
\pgfusepath{stroke}%
\end{pgfscope}%
\begin{pgfscope}%
\pgfsetrectcap%
\pgfsetmiterjoin%
\pgfsetlinewidth{0.803000pt}%
\definecolor{currentstroke}{rgb}{0.000000,0.000000,0.000000}%
\pgfsetstrokecolor{currentstroke}%
\pgfsetdash{}{0pt}%
\pgfpathmoveto{\pgfqpoint{5.900712in}{0.494721in}}%
\pgfpathlineto{\pgfqpoint{5.900712in}{3.575133in}}%
\pgfusepath{stroke}%
\end{pgfscope}%
\begin{pgfscope}%
\pgfsetrectcap%
\pgfsetmiterjoin%
\pgfsetlinewidth{0.803000pt}%
\definecolor{currentstroke}{rgb}{0.000000,0.000000,0.000000}%
\pgfsetstrokecolor{currentstroke}%
\pgfsetdash{}{0pt}%
\pgfpathmoveto{\pgfqpoint{0.630692in}{0.494721in}}%
\pgfpathlineto{\pgfqpoint{5.900712in}{0.494721in}}%
\pgfusepath{stroke}%
\end{pgfscope}%
\begin{pgfscope}%
\pgfsetrectcap%
\pgfsetmiterjoin%
\pgfsetlinewidth{0.803000pt}%
\definecolor{currentstroke}{rgb}{0.000000,0.000000,0.000000}%
\pgfsetstrokecolor{currentstroke}%
\pgfsetdash{}{0pt}%
\pgfpathmoveto{\pgfqpoint{0.630692in}{3.575133in}}%
\pgfpathlineto{\pgfqpoint{5.900712in}{3.575133in}}%
\pgfusepath{stroke}%
\end{pgfscope}%
\begin{pgfscope}%
\pgfsetbuttcap%
\pgfsetmiterjoin%
\definecolor{currentfill}{rgb}{1.000000,1.000000,1.000000}%
\pgfsetfillcolor{currentfill}%
\pgfsetfillopacity{0.800000}%
\pgfsetlinewidth{1.003750pt}%
\definecolor{currentstroke}{rgb}{0.800000,0.800000,0.800000}%
\pgfsetstrokecolor{currentstroke}%
\pgfsetstrokeopacity{0.800000}%
\pgfsetdash{}{0pt}%
\pgfpathmoveto{\pgfqpoint{4.871403in}{0.557221in}}%
\pgfpathlineto{\pgfqpoint{5.813212in}{0.557221in}}%
\pgfpathquadraticcurveto{\pgfqpoint{5.838212in}{0.557221in}}{\pgfqpoint{5.838212in}{0.582221in}}%
\pgfpathlineto{\pgfqpoint{5.838212in}{1.120135in}}%
\pgfpathquadraticcurveto{\pgfqpoint{5.838212in}{1.145135in}}{\pgfqpoint{5.813212in}{1.145135in}}%
\pgfpathlineto{\pgfqpoint{4.871403in}{1.145135in}}%
\pgfpathquadraticcurveto{\pgfqpoint{4.846403in}{1.145135in}}{\pgfqpoint{4.846403in}{1.120135in}}%
\pgfpathlineto{\pgfqpoint{4.846403in}{0.582221in}}%
\pgfpathquadraticcurveto{\pgfqpoint{4.846403in}{0.557221in}}{\pgfqpoint{4.871403in}{0.557221in}}%
\pgfpathlineto{\pgfqpoint{4.871403in}{0.557221in}}%
\pgfpathclose%
\pgfusepath{stroke,fill}%
\end{pgfscope}%
\begin{pgfscope}%
\pgfsetbuttcap%
\pgfsetroundjoin%
\pgfsetlinewidth{1.003750pt}%
\definecolor{currentstroke}{rgb}{0.121569,0.466667,0.705882}%
\pgfsetstrokecolor{currentstroke}%
\pgfsetdash{{3.700000pt}{1.600000pt}}{0.000000pt}%
\pgfpathmoveto{\pgfqpoint{4.896403in}{1.043915in}}%
\pgfpathlineto{\pgfqpoint{5.021403in}{1.043915in}}%
\pgfpathlineto{\pgfqpoint{5.146403in}{1.043915in}}%
\pgfusepath{stroke}%
\end{pgfscope}%
\begin{pgfscope}%
\definecolor{textcolor}{rgb}{0.000000,0.000000,0.000000}%
\pgfsetstrokecolor{textcolor}%
\pgfsetfillcolor{textcolor}%
\pgftext[x=5.246403in,y=1.000165in,left,base]{\color{textcolor}\sffamily\fontsize{9.000000}{10.800000}\selectfont DC-SSAT}%
\end{pgfscope}%
\begin{pgfscope}%
\pgfsetbuttcap%
\pgfsetroundjoin%
\pgfsetlinewidth{1.003750pt}%
\definecolor{currentstroke}{rgb}{1.000000,0.498039,0.054902}%
\pgfsetstrokecolor{currentstroke}%
\pgfsetdash{{3.700000pt}{1.600000pt}}{0.000000pt}%
\pgfpathmoveto{\pgfqpoint{4.896403in}{0.860443in}}%
\pgfpathlineto{\pgfqpoint{5.021403in}{0.860443in}}%
\pgfpathlineto{\pgfqpoint{5.146403in}{0.860443in}}%
\pgfusepath{stroke}%
\end{pgfscope}%
\begin{pgfscope}%
\definecolor{textcolor}{rgb}{0.000000,0.000000,0.000000}%
\pgfsetstrokecolor{textcolor}%
\pgfsetfillcolor{textcolor}%
\pgftext[x=5.246403in,y=0.816693in,left,base]{\color{textcolor}\sffamily\fontsize{9.000000}{10.800000}\selectfont erSSAT}%
\end{pgfscope}%
\begin{pgfscope}%
\pgfsetrectcap%
\pgfsetroundjoin%
\pgfsetlinewidth{1.003750pt}%
\definecolor{currentstroke}{rgb}{0.172549,0.627451,0.172549}%
\pgfsetstrokecolor{currentstroke}%
\pgfsetdash{}{0pt}%
\pgfpathmoveto{\pgfqpoint{4.896403in}{0.676972in}}%
\pgfpathlineto{\pgfqpoint{5.021403in}{0.676972in}}%
\pgfpathlineto{\pgfqpoint{5.146403in}{0.676972in}}%
\pgfusepath{stroke}%
\end{pgfscope}%
\begin{pgfscope}%
\definecolor{textcolor}{rgb}{0.000000,0.000000,0.000000}%
\pgfsetstrokecolor{textcolor}%
\pgfsetfillcolor{textcolor}%
\pgftext[x=5.246403in,y=0.633222in,left,base]{\color{textcolor}\sffamily\fontsize{9.000000}{10.800000}\selectfont DPER}%
\end{pgfscope}%
\end{pgfpicture}%
\makeatother%
\endgroup%

%% file: dper.bbl
\begin{thebibliography}{}

\bibitem[Abo~Khamis et~al., 2016]{abo2016faq}
Abo~Khamis, M., Ngo, H.~Q., and Rudra, A. (2016).
\newblock {FAQ: questions asked frequently}.
\newblock In {\em PODS}.

\bibitem[Bahar et~al., 1997]{bahar1997algebraic}
Bahar, R.~I., Frohm, E.~A., Gaona, C.~M., Hachtel, G.~D., Macii, E., Pardo, A.,
  and Somenzi, F. (1997).
\newblock {Algebraic decision diagrams and their applications}.
\newblock {\em Form Method Syst Des}.

\bibitem[Balint et~al., 2015]{balint2015overview}
Balint, A., Belov, A., Jarvisalo, M., and Sinz, C. (2015).
\newblock {Overview and analysis of the SAT Challenge 2012 solver competition}.
\newblock {\em AIJ}.

\bibitem[Bellman, 1966]{bellman1966dynamic}
Bellman, R. (1966).
\newblock {Dynamic programming}.
\newblock {\em Science}.

\bibitem[Bryant, 1986]{bryant1986graph}
Bryant, R.~E. (1986).
\newblock {Graph-based algorithms for Boolean function manipulation}.
\newblock {\em IEEE TC}.

\bibitem[Cai et~al., 2017]{cai2017bayesian}
Cai, B., Huang, L., and Xie, M. (2017).
\newblock {Bayesian networks in fault diagnosis}.
\newblock {\em IEEE Trans Industr Inform}.

\bibitem[Chakraborty et~al., 2016]{chakraborty2016algorithmic}
Chakraborty, S., Meel, K.~S., and Vardi, M.~Y. (2016).
\newblock {Algorithmic improvements in approximate counting for probabilistic
  inference: from linear to logarithmic SAT calls}.
\newblock In {\em IJCAI}.

\bibitem[Charwat and Woltran, 2016]{charwat2016bdd}
Charwat, G. and Woltran, S. (2016).
\newblock {BDD-based dynamic programming on tree decompositions}.
\newblock Technical report, Technische Universitat Wien.

\bibitem[Chavira et~al., 2005]{chavira2005exploiting}
Chavira, M., Allen, D., and Darwiche, A. (2005).
\newblock {Exploiting evidence in probabilistic inference}.
\newblock In {\em UAI}.

\bibitem[Cook, 1971]{cook1971complexity}
Cook, S.~A. (1971).
\newblock {The complexity of theorem-proving procedures}.
\newblock In {\em STOC}.

\bibitem[Crama et~al., 1990]{crama1990basic}
Crama, Y., Hansen, P., and Jaumard, B. (1990).
\newblock {The basic algorithm for pseudo-Boolean programming revisited}.
\newblock {\em Discrete Applied Mathematics}.

\bibitem[Dudek et~al., 2020]{dudek2020dpmc}
Dudek, J.~M., Phan, V. H.~N., and Vardi, M.~Y. (2020).
\newblock {DPMC: weighted model counting by dynamic programming on project-join
  trees}.
\newblock In {\em CP}.

\bibitem[Dudek et~al., 2021]{dudek2021procount}
Dudek, J.~M., Phan, V. H.~N., and Vardi, M.~Y. (2021).
\newblock {ProCount: weighted projected model counting with graded project-join
  trees}.
\newblock In {\em SAT}.

\bibitem[Fichte et~al., 2020]{fichte2020exploiting}
Fichte, J.~K., Hecher, M., Thier, P., and Woltran, S. (2020).
\newblock {Exploiting database management systems and treewidth for counting}.
\newblock In {\em PADL}.

\bibitem[Fremont et~al., 2017]{fremont2017maximum}
Fremont, D.~J., Rabe, M.~N., and Seshia, S.~A. (2017).
\newblock {Maximum model counting}.
\newblock In {\em AAAI}.

\bibitem[Froleyks et~al., 2021]{froleyks2021sat}
Froleyks, N., Heule, M., Iser, M., Jarvisalo, M., and Suda, M. (2021).
\newblock {SAT Competition 2020}.
\newblock {\em AIJ}.

\bibitem[Ghosh et~al., 2021]{ghosh2021justicia}
Ghosh, B., Basu, D., and Meel, K.~S. (2021).
\newblock {Justicia: a stochastic SAT approach to formally verify fairness}.
\newblock In {\em AAAI}.

\bibitem[Gupta et~al., 2019]{gupta2019waps}
Gupta, R., Sharma, S., Roy, S., and Meel, K.~S. (2019).
\newblock {WAPS: weighted and projected sampling}.
\newblock In {\em TACAS}.

\bibitem[Janota and Marques-Silva, 2015]{janota2015solving}
Janota, M. and Marques-Silva, J. (2015).
\newblock {Solving QBF by clause selection}.
\newblock In {\em IJCAI}.

\bibitem[Krentel, 1988]{krentel1988complexity}
Krentel, M.~W. (1988).
\newblock {The complexity of optimization problems}.
\newblock {\em Journal of computer and system sciences}.

\bibitem[Kyrillidis et~al., 2022]{kyrillidis2022dpms}
Kyrillidis, A., Vardi, M.~Y., and Zhang, Z. (2022).
\newblock {DPMS: an ADD-based symbolic approach for generalized MaxSAT
  solving}.
\newblock {\em arXiv preprint arXiv:2205.03747}.

\bibitem[Lee et~al., 2017]{lee2017solving}
Lee, N.-Z., Wang, Y.-S., and Jiang, J.-H.~R. (2017).
\newblock {Solving stochastic Boolean satisfiability under random-exist
  quantification}.
\newblock In {\em IJCAI}.

\bibitem[Lee et~al., 2018]{lee2018solving}
Lee, N.-Z., Wang, Y.-S., and Jiang, J.-H.~R. (2018).
\newblock {Solving exist-random quantified stochastic Boolean satisfiability
  via clause selection}.
\newblock In {\em IJCAI}.

\bibitem[Littman et~al., 2001]{littman2001stochastic}
Littman, M.~L., Majercik, S.~M., and Pitassi, T. (2001).
\newblock {Stochastic Boolean satisfiability}.
\newblock {\em J Autom Reasoning}.

\bibitem[Majercik and Boots, 2005]{majercik2005dc}
Majercik, S.~M. and Boots, B. (2005).
\newblock {DC-SSAT: a divide-and-conquer approach to solving stochastic
  satisfiability problems efficiently}.
\newblock In {\em AAAI}.

\bibitem[Pan and Vardi, 2005]{pan2005symbolic}
Pan, G. and Vardi, M.~Y. (2005).
\newblock {Symbolic techniques in satisfiability solving}.
\newblock {\em J Autom Reasoning}.

\bibitem[Papadimitriou, 1985]{papadimitriou1985games}
Papadimitriou, C.~H. (1985).
\newblock {Games against nature}.
\newblock {\em Journal of Computer and System Sciences}.

\bibitem[Park, 2002]{park2002using}
Park, J.~D. (2002).
\newblock {Using weighted MaxSAT engines to solve MPE}.
\newblock In {\em AAAI/IAAI}.

\bibitem[Park and Darwiche, 2004]{park2004complexity}
Park, J.~D. and Darwiche, A. (2004).
\newblock {Complexity results and approximation strategies for MAP
  explanations}.
\newblock {\em JAIR}.

\bibitem[Pearl, 1985]{pearl1985bayesian}
Pearl, J. (1985).
\newblock {Bayesian networks: a model cf self-activated memory for evidential
  reasoning}.
\newblock In {\em CogSci}.

\bibitem[Phan and Vardi, 2022]{phan2022dpo}
Phan, V. H.~N. and Vardi, M.~Y. (2022).
\newblock {DPO: dynamic-programming optimization on hybrid constraints}.
\newblock {\em arXiv preprint arXiv:2205.08632}.

\bibitem[Robertson and Seymour, 1991]{robertson1991graph}
Robertson, N. and Seymour, P.~D. (1991).
\newblock {Graph minors. X. Obstructions to tree-decomposition}.
\newblock {\em J Combinatorial Theory B}.

\bibitem[Sang et~al., 2005]{sang2005performing}
Sang, T., Beame, P., and Kautz, H.~A. (2005).
\newblock {Performing Bayesian inference by weighted model counting}.
\newblock In {\em AAAI}.

\bibitem[Shwe et~al., 1991]{shwe1991probabilistic}
Shwe, M.~A., Middleton, B., Heckerman, D.~E., Henrion, M., Horvitz, E.~J.,
  Lehmann, H.~P., and Cooper, G.~F. (1991).
\newblock {Probabilistic diagnosis using a reformulation of the INTERNIST-1/QMR
  knowledge base}.
\newblock {\em Methods of information in Medicine}.

\bibitem[Somenzi, 2015]{somenzi2015cudd}
Somenzi, F. (2015).
\newblock {\em {CUDD: CU decision diagram package--release 3.0.0}}.

\bibitem[Soos and Meel, 2019]{soos2019bird}
Soos, M. and Meel, K.~S. (2019).
\newblock {BIRD: engineering an efficient CNF-XOR SAT solver and its
  applications to approximate model counting}.
\newblock In {\em AAAI}.

\bibitem[Strasser, 2017]{strasser2017computing}
Strasser, B. (2017).
\newblock {Computing tree decompositions with FlowCutter: PACE 2017
  submission}.
\newblock {\em arXiv preprint arXiv:1709.08949}.

\bibitem[Valiant, 1979]{valiant1979complexity}
Valiant, L.~G. (1979).
\newblock {The complexity of enumeration and reliability problems}.
\newblock {\em SICOMP}.

\bibitem[van Dijk and van~de Pol, 2015]{van2015sylvan}
van Dijk, T. and van~de Pol, J. (2015).
\newblock {Sylvan: multi-core decision diagrams}.
\newblock In {\em TACAS}.

\bibitem[Xu et~al., 2008]{xu2008satzilla}
Xu, L., Hutter, F., Hoos, H.~H., and Leyton-Brown, K. (2008).
\newblock {SATzilla: portfolio-based algorithm selection for SAT}.
\newblock {\em JAIR}.

\bibitem[Zawadzki et~al., 2013]{zawadzki2013generalization}
Zawadzki, E.~P., Platzer, A., and Gordon, G.~J. (2013).
\newblock {A generalization of SAT and \#SAT for robust policy evaluation}.
\newblock In {\em IJCAI}.

\end{thebibliography}
